\documentclass[12pt]{article}
\usepackage{amssymb,amsmath,amsthm,amscd,latexsym}
\usepackage{amsfonts}
\usepackage{tikz}
\usetikzlibrary{arrows,positioning,shapes,fit,calc}
\usepackage{tikz-cd}
\usepackage[mathscr]{eucal}
\usepackage{color,hyperref}
\usepackage{graphicx}
\usepackage{multirow}

\newtheorem{thm}{Theorem}[section]
\newtheorem{prop}[thm]{Proposition}
\newtheorem{lem}[thm]{Lemma}
\newtheorem{cor}[thm]{Corollary}
\newtheorem{exam}{Example}

\def\1v{\mathbf{1}}
\def\0v{\mathbf{0}}

\begin{document}
\title{Composite Matrices from Group Rings, Composite $G$-Codes and Constructions of Self-Dual Codes}
\author{ Steven T. Dougherty \\
Department of Mathematics \\
University of Scranton \\
Scranton, PA 18510 \\
USA \\
Joe Gildea, Adrian Korban \\
Department of Mathematical and Physical Sciences \\
University of Chester \\
Thornton Science Park, Pool Ln, Chester CH2 4NU, England\\
Abidin Kaya \\
Department of Mathematics Education\\ Sampoerna University,
12780, Jakarta, Indonesia}
\maketitle

\begin{abstract}
In this work, we define  composite matrices which are derived from group rings. We extend the idea of $G$-codes to composite $G$-codes. We show that these codes are ideals in a group ring, where the ring is a finite commutative Frobenius ring and $G$ is an arbitrary finite group. We prove that the dual of a composite $G$-code is also a composite $G$-code. We define quasi-composite $G$-codes and give a construction of these codes. We also study generator matrices, which consist of the identity matrices and the composite matrices. Together with the generator matrices, the well known extension method, the neighbour method and its generalization, we find extremal binary self-dual codes of length 68 with new weight enumerators for the rare parameters $\gamma=7,8$ and $9.$ In particular, we find 49 new such codes.   Moreover, we show that the codes we find are inaccessible from other constructions.  
\end{abstract}


\section{Introduction}

Recently, the authors in \cite{XIX} have defined $G$-codes which are ideals in a group ring, where the ring is a finite commutative Frobenius ring and $G$ is an arbitrary finite group. This idea is based on applying the matrix $\sigma(v),$ where $v$ is a group ring element,  which was first introduced in \cite{XX}. In another recent paper \cite{XII}, the authors also apply the matrix $\sigma(v)$ to form generator matrices of the form $[I \ | \ \sigma(v)],$ where $I$ is the identity matrix, to search for extremal self-dual codes. 

In \cite{XXI} and \cite{XXII}, the authors amend the generator matrices of the form $[I \ | \ \sigma(v)].$ Namely, they replace the matrices $\sigma(v)$ with more complex matrices and call these the composite constructions. The authors define a number of composite constructions which  they apply to search for extremal self-dual codes, but no general theory is given. Later in \cite{I}, the authors generalize the idea of composite constructions, amend the matrix $\sigma(v)$ and define the composite constructions, which they also call the composite matrices, in a more general and rigorous way.

In this work, we look at the composite matrices and compare them to the matrices obtained from $\sigma(v).$ We look at when these matrices are equivalent and when they are not. We next apply the more general and rigorous definition of the composite matrices to define  composite $G$-codes. This is an extension of $G$-codes mentioned earlier. We show that these codes are ideals in the group ring, and that the dual of a composite $G$-code is also a composite $G$-code. Moreover, we show when the composite $G$-codes are self-orthogonal and self-dual. Additionally, we define quasi-composite $G$-codes and extend some known results to quasi $G$-codes.

We generalize the theory of the composite constructions defined in \cite{XXI} and \cite{XXII}. Namely, we show in general when and under what conditions such matrices produce self-dual codes, rather than showing it for individual cases as in \cite{XXI} and \cite{XXII}. Lastly in this paper, we combine the ideas of composite matrices, the well known extension method and the neighbour construction and its generalization (see \cite{GN} for details), to search for extremal binary self-dual codes of length 68. As a result, we obtain 49 such codes with parameters that were not known in the literature before.

The rest of the work is organized as follows. In Section 2, we give preliminary definitions and results on codes, group rings and special matrices. We also recall the construction of $G$-codes from \cite{XIX}. In Section 3, we define the composite matrix which was introduced in \cite{I} and compare it with the matrix $\sigma(v).$ In Section 4, we define the composite $G$-codes and show that these are ideals in the group ring. In Section 5, we study when the composite $G$-codes are self-orthogonal and self-dual. Section 6 consists of a study of quasi-composite $G$-codes. In Section 7, we generalize the theory of the composite constructions from \cite{XXI} and \cite{XXII}. Additionally, we find new extremal self-dual binary codes of length 68. We finish with concluding remarks and directions for possible future research.

\section{Preliminaries}

\subsection{Codes, Group Rings and Special Matrices}

We begin by recalling the standard definitions from coding theory. In this paper, all rings are assumed to be commutative, finite, Frobenius rings with a multiplicative identity. Denote the character module of $R$ by $\widehat{R}.$ A code $C$
of length $n$ over a Frobenius ring $R$ is a subset of $R^n$. For a finite ring $R$ the following
are equivalent:

\begin{enumerate}
\item[(1)] $R$ is a Frobenius ring;

\item[(2)] As a left module, $\widehat{R} \cong {\ }_RR$;

\item[(3)] As a right module $\widehat{R} \cong R_R.$
\end{enumerate}

We consider codes over Frobenius rings since such rings have good duality properties which are reflected by the equivalent statements above. If the code is a submodule of $R^n$, then we say that the code is linear. For
a full description of Frobenius rings and codes over Frobenius rings, see
\cite{XVI}. Elements of the code $C$ are called codewords of $C$. Let $
\mathbf{x}=(x_1,x_2,\dots,x_n)$ and $\mathbf{y}=(y_1,y_2,\dots,y_n)$ be two
elements of $R^n.$ The duality is understood in terms of the Euclidean inner
product, namely:
\begin{equation*}
\langle \mathbf{x},\mathbf{y} \rangle_E=\sum x_iy_i.
\end{equation*}
The dual $C^{\bot}$ of the code $C$ is defined as
\begin{equation*}
C^{\bot}=\{\mathbf{x} \in R^n \ | \ \langle \mathbf{x},\mathbf{y}
\rangle_E=0 \ \text{for all} \ \mathbf{y} \in C\}.
\end{equation*}
We say that $C$ is self-orthogonal if $C \subseteq C^\perp$ and is self-dual
if $C=C^{\bot}.$

We next recall the standard definitions and notations for group rings. Let $RG$ denote the group ring of the group $G$ over the ring $R.$ A non-zero element $z$ in a ring $R$ is said to be a zero-divisor in $R$ if and only if there exists a non-zero element $r \in R$ with $z \ast r=0.$ When $R$ has an identity $1_R,$ we say $u$ is a unit in $R$ if and only if there exists an element $w \in R$ with $u \ast w=1_R.$ The group of units of $R$ is denoted by $U(R).$ Let $R_{n \times n}$ denote the ring of $n \times n$ matrices with coefficients from $R.$ While group rings can be given for
infinite rings and infinite groups, we are only concerned with group rings
where both the ring and the group are finite. Let $G$ be a finite group of
order $n$, then the group ring $RG$ consists of $\sum_{i=1}^n \alpha_i g_i$,
$\alpha_i \in R$, $g_i \in G.$

Addition in the group ring is done by coordinate addition, namely
\begin{equation}
\sum_{i=1}^n \alpha_i g_i +\sum_{i=1}^n \beta_i g_i =\sum_{i=1}^n (\alpha_i
+ \beta_i)g_i.
\end{equation}
The product of two elements in a group ring is given by
\begin{equation}
\left(\sum_{i=1}^n \alpha_i g_i \right)\left(\sum_{j=1}^n \beta_j g_j
\right)= \sum_{i,j} \alpha_i \beta_j g_i g_j.
\end{equation}
It follows that the coefficient of $g_k$ in the product is $\sum_{g_i
g_j=g_k} \alpha_i \beta_j.$ For more details on group rings,  see \cite{XVII} and \cite{XVIII}.

A right circulant matrix is one where each row is shifted one element to the right relative to the preceding row. Since we shall always shift to the right in this work, we shall simply call it a circulant matrix.  We label the circulant matrix as $A=circ(\alpha_1,\alpha_2,\dots,\alpha_n),$ where $\alpha_i$ are rings elements. The transpose of a matrix $A,$ denoted by $A^T,$ is a matrix whose rows are the columns of $A,$ that is $A_{ij}^T=A_{ji}.$

\subsection{$G$-Codes}

The following construction of a matrix was first given for codes over fields
by Hurley in \cite{XX}. It was extended to Frobenius rings in \cite{XIX}. Let $%
R$ be a finite commutative Frobenius ring and let $G=\{g_1,g_2,\dots,g_n\}$
be a group of order $n$. Let $v=\alpha_{g_1}g_1+\alpha_{g_2}g_2+\dots
+\alpha_{g_n}g_n \in RG.$ Define the matrix $\sigma(v) \in M_n(R)$ to be
\begin{equation}
\sigma(v)=%
\begin{pmatrix}
\alpha_{g_1^{-1}g_1} & \alpha_{g_1^{-1}g_2} & \alpha_{g_1^{-1}g_3} & \dots &
\alpha_{g_1^{-1}g_n} \\
\alpha_{g_2^{-1}g_1} & \alpha_{g_2^{-1}g_2} & \alpha_{g_2^{-1}g_3} & \dots &
\alpha_{g_2^{-1}g_n} \\
\vdots & \vdots & \vdots & \vdots & \vdots \\
\alpha_{g_n^{-1}g_1} & \alpha_{g_n^{-1}g_2} & \alpha_{g_n^{-1}g_3} & \dots &
\alpha_{g_n^{-1}g_n}%
\end{pmatrix}%
.
\end{equation}

We note that the elements $g_1^{-1}, g_2^{-1}, \dots, g_n^{-1}$ are the elements of the group $G$ in a some given order. For a given element $v \in RG,$ we define the $G$-code over the ring $R:$

\begin{equation}
\mathcal{C}(v)=\langle \sigma(v) \rangle,
\end{equation}
where this indicates that the code is formed by taking the row space of $\sigma(v)$ over the ring $R.$
It has been shown that $\mathcal{C}(v)$ corresponds to an ideal in the group ring $RG$.  

\section{The Composite $\Omega(v)$ matrix}

In this section, we define the composite matrix $\Omega(v)$ which was first introduced in \cite{I} and compare it with the matrix $\sigma(v).$  

Let $R$ be a finite commutative Frobenius ring. Let $\{g_1,g_2,\dots ,g_n\}$ be a fixed listing of the elements of $G.$ Let $\{(h_i)_1,(h_i)_2,\dots ,(h_i)_r\}$  be a fixed listing of the elements of $H_i,$ where $H_i$ is any group of order $r.$ Let $r$ be a factor of $n$ with $n>r$ and $n,r \neq 1.$ Also, let $G_r$ be a subset of $G$ containing $r$ distinct elements of $G.$ Define the following  map:

\begin{table}[h!]
\centering
\begin{tabular}{ccc}
\multicolumn{3}{c}{$\phi: H_i \mapsto G_r$}                \\
$(h_i)_1$ & $\xrightarrow{\phi}$ & $g_{j^{-1}}g_k$         \\
$(h_i)_2$ & $\xrightarrow{\phi}$ & $g_{j^{-1}}g_{k+1}$     \\
$\vdots$  & $\vdots$             & $\vdots$                \\
$(h_i)_r$ & $\xrightarrow{\phi}$ & $g_{j^{-1}}g_{k+(r-1)}.$
\end{tabular}
\end{table}

It was shown in \cite{I} that the map $\phi$ is a bijection.

Let $v=\alpha_{g_1}g_1+\alpha_{g_2}g_2+\dots,\alpha_{g_n}g_n \in RG.$ Define the matrix $\Omega(v) \in M_n(R)$ to be

\begin{equation}\label{Composite construction}
\Omega(v)=\begin{pmatrix}
A_1&A_2&A_3&\dots &A_{\frac{n}{r}}\\
A_{\frac{n}{r}+1}&A_{\frac{n}{r}+2}&A_{\frac{n}{r}+3}&\dots &A_{\frac{2n}{r}}\\
\vdots & \vdots & \vdots & \vdots & \vdots \\
A_{\frac{(r-1)n}{r}+1}&A_{\frac{(r-1)n}{r}+2}&A_{\frac{(r-1)n}{r}+3}&\dots & A_{\frac{n^2}{r^2}}
\end{pmatrix},
\end{equation}
where at least one block has the following form:
$$A_l=\begin{pmatrix}
\alpha_{g_j^{-1}g_k}&\alpha_{g_j^{-1}g_{k+1}}&\dots &\alpha_{g_j^{-1}g_{k+(r-1)}}\\
\alpha_{g_{j+1}^{-1}g_k}&\alpha_{g_{j+1}^{-1}g_{k+1}}&\dots &\alpha_{g_{j+1}^{-1}g_{k+(r-1)}}\\
\alpha_{g_{j+2}^{-1}g_k}&\alpha_{g_{j+2}^{-1}g_{k+1}}&\dots &\alpha_{g_{j+2}^{-1}g_{k+(r-1)}}\\
\vdots & \vdots & \vdots & \vdots\\
\alpha_{g_{j+r-1}^{-1}g_k}&\alpha_{g_{j+r-1}^{-1}g_{k+1}}&\dots &\alpha_{g_{j+r-1}^{-1}g_{k+(r-1)}}
\end{pmatrix},$$
and the other blocks are of the form:
$$A_l'=\begin{pmatrix}
\alpha_{g_j^{-1}g_k}&\alpha_{g_j^{-1}g_{k+1}}&\dots &\alpha_{g_j^{-1}g_{k+(r-1)}}\\
\alpha_{\phi_l((h_i)_2^{-1}(h_i)_1)}&\alpha_{\phi_l((h_i)_2^{-1}(h_i)_2)}&\dots &\alpha_{\phi_l((h_i)_2^{-1}(h_i)_r)}\\
\alpha_{\phi_l((h_i)_3^{-1}(h_i)_1)}&\alpha_{\phi_l((h_i)_3^{-1}(h_i)_2)}&\dots &\alpha_{\phi_l((h_i)_3^{-1}(h_i)_r)}\\
\vdots & \vdots & \vdots & \vdots\\
\alpha_{\phi_l((h_i)_r^{-1}(h_i)_1)}&\alpha_{\phi_l((h_i)_r^{-1}(h_i)_2)}&\dots &\alpha_{\phi_l((h_i)_r^{-1}(h_i)_r)}\\
\end{pmatrix},$$
where in both cases, when $l=1$ then $j=1,k=1,$ when $l=2$ then $j=1,k=r+1,$ when $l=3$ then $j=1,k=2r+1,$ $\dots$ when $l=\frac{n}{r}$ then $j=1,k=n-r+1.$ When $l=\frac{n}{r}+1$ then $j=r+1, k=1,$ when $l=\frac{n}{r}+2$ then $j=r+1, k=r+1,$ when $l=\frac{n}{r}+3$ then $j=r+1, k=2r+1,$ $\dots$ when $l=\frac{2n}{r}$ then $j=r+1, k=n-r+1,$ $\dots,$ and so on.

We note that if the above matrix $\Omega(v)$ consists of blocks which are of the $A_l$ form only, then it is the same as the matrix $\sigma(v)$ from \cite{XX}. Therefore, from now on  we assume that the matrix $\Omega(v)$ consists of at least one block of the $A_l'$ form. It is also clear that the matrix $\Omega(v)$ cannot be constructed when the order of the group $G$ is odd. In each block, the first row consists of $r$ distinct elements of $G.$ The map $\phi_l$ is applied in individual blocks which means we can employ $\frac{n^2}{r^2}$ different maps $\phi_l$ and $\frac{n^2}{r^2}$ different groups of order $r$ (if that many exist). This is the advantage of our construction over the matrix $\sigma(v)$, namely, by employing different groups of order $r$ and by applying the maps $\phi_l$ in individual blocks, we construct more complex matrices over the ring $R.$ We call the matrix $\Omega(v)$ the composite $G$-matrix.

The rows of the matrix $\sigma(v)$ in \cite{XX} consist of the vectors that correspond to the elements $hv$ in $RG$ where $h$ is any element of $G.$ This is not the case in the composite matrix $\Omega(v).$    

\begin{exam}\label{ex.2}
Let $G=\langle x,y \ | \ x^4=y^2=1, x^y=x^{-1} \rangle \cong D_8.$ Let $v=\alpha_1+\alpha_{x}x+\alpha_{x^2}x^2+\alpha_{x^3}x^3+\alpha_{y}y+\alpha_{xy}xy+\alpha_{x^2y}x^2y+\alpha_{x^3y}x^3y \in RD_8,$ where $\alpha_{g_i} \in R.$ Let $H_1=\langle a,b \ | \ a^2=b^2=1, ab=ba \rangle \cong C_2 \times C_2.$ We now define the composite matrix as:
$$\Omega(v)=\begin{pmatrix}
A_1'&A_2'\\
A_3&A_4
\end{pmatrix}=$$

\resizebox{0.7\textwidth}{!}{\begin{minipage}{\textwidth}
$$\begin{pmatrix}
\begin{tabular}{cccc|cccc}
$\alpha_{{g_1^{-1}g_1}}$ & $\alpha_{{g_1^{-1}g_2}}$ & $\alpha_{{g_1^{-1}g_3}}$ & $\alpha_{{g_1^{-1}g_4}}$ & $\alpha_{{g_1^{-1}g_5}}$ & $\alpha_{{g_1^{-1}g_6}}$ & $\alpha_{{g_1^{-1}g_7}}$ & $\alpha_{{g_1^{-1}g_8}}$ \\
$\alpha_{\phi_1((h_1)_2^{-1}(h_1)_1)}$                        & $\alpha_{\phi_1((h_1)_2^{-1}(h_1)_2)}$                        & $\alpha_{\phi_1((h_1)_2^{-1}(h_1)_3)}$                        & $\alpha_{\phi_1((h_1)_2^{-1}(h_1)_4)}$                        & $\alpha_{\phi_2((h_1)_2^{-1}(h_1)_1)}$                        & $\alpha_{\phi_2((h_1)_2^{-1}(h_1)_2)}$                        & $\alpha_{\phi_2((h_1)_2^{-1}(h_1)_3)}$                        & $\alpha_{\phi_2((h_1)_2^{-1}(h_1)_4)}$                        \\
$\alpha_{\phi_1((h_1)_3^{-1}(h_1)_1)}$                        & $\alpha_{\phi_1((h_1)_3^{-1}(h_1)_2)}$                        & $\alpha_{\phi_1(h_1)_3^{-1}(h_1)_3)}$                        & $\alpha_{\phi_1((h_1)_3^{-1}(h_1)_4)}$                        & $\alpha_{\phi_2((h_1)_3^{-1}(h_1)_1)}$                        & $\alpha_{\phi_2((h_1)_3^{-1}(h_1)_2)}$                        & $\alpha_{\phi_2((h_1)_3^{-1}(h_1)_3)}$                        & $\alpha_{\phi_2((h_1)_3^{-1}(h_1)_4)}$                        \\
$\alpha_{\phi_1((h_1)_4^{-1}(h_1)_1)}$                        & $\alpha_{\phi_1((h_1)_4^{-1}(h_1)_2)}$                        & $\alpha_{\phi_1(h_1)_4^{-1}(h_1)_3)}$                        & $\alpha_{\phi_1((h_1)_4^{-1}(h_1)_4)}$                        & $\alpha_{\phi_2((h_1)_4^{-1}(h_1)_1)}$                        & $\alpha_{\phi_2((h_1)_4^{-1}(h_1)_2)}$                        & $\alpha_{\phi_2((h_1)_4^{-1}(h_1)_3)}$                        & $\alpha_{\phi_2((h_1)_4^{-1}(h_1)_4)}$                        \\ \hline
$\alpha_{{g_5^{-1}g_1}}$ & $\alpha_{{g_5^{-1}g_2}}$ & $\alpha_{{g_5^{-1}g_3}}$ & $\alpha_{{g_5^{-1}g_4}}$ & $\alpha_{{g_5^{-1}g_5}}$ & $\alpha_{{g_5^{-1}g_6}}$ & $\alpha_{{g_5^{-1}g_7}}$ & $\alpha_{{g_5^{-1}g_8}}$ \\
$\alpha_{{g_6^{-1}g_1}}$                        & $\alpha_{{g_6^{-1}g_2}}$                        & $\alpha_{{g_6^{-1}g_3}}$                        & $\alpha_{{g_6^{-1}g_4}}$                        & $\alpha_{{g_6^{-1}g_5}}$                        & $\alpha_{{g_6^{-1}g_6}}$                        & $\alpha_{{g_6^{-1}g_7}}$                        & $\alpha_{{g_6^{-1}g_8}}$                                               \\
$\alpha_{{g_7^{-1}g_1}}$                        & $\alpha_{{g_7^{-1}g_2}}$                        & $\alpha_{{g_7^{-1}g_3}}$                        & $\alpha_{{g_7^{-1}g_4}}$                        & $\alpha_{{g_7^{-1}g_5}}$                        & $\alpha_{{g_7^{-1}g_6}}$                        & $\alpha_{{g_7^{-1}g_7}}$                        & $\alpha_{{g_7^{-1}g_8}}$                                               \\
$\alpha_{{g_8^{-1}g_1}}$                        & $\alpha_{{g_8^{-1}g_2}}$                        & $\alpha_{{g_8^{-1}g_3}}$                        & $\alpha_{{g_8^{-1}g_4}}$                        & $\alpha_{{g_8^{-1}g_5}}$                        & $\alpha_{{g_8^{-1}g_6}}$                        & $\alpha_{{g_8^{-1}g_7}}$                        & $\alpha_{{g_8^{-1}g_8}}$                                              
\end{tabular}
\end{pmatrix},$$
\end{minipage}}\\

where:
\begin{table}[h!]
\centering
\begin{tabular}{cclcc}
\multirow{2}{*}{$\phi_1:$} & $(h_1)_i \xrightarrow{\phi_1} g_1^{-1}g_i$ &  & \multirow{2}{*}{$\phi_2:$} & $(h_1)_i \xrightarrow{\phi_2} g_1^{-1}g_j$    \\
                           & $\text{for} \ i=\{1,2,3,4\}$               &  &                            & $\text{for when} \ \{i=1,j=5,i=2,j=6,i=3,j=7,i=4,j=8\}$ 
\end{tabular}
\end{table}

in $A_1'$ and $A_2'.$ This results in a composite matrix over $R$ of the following form:

$$\Omega(v)=\begin{pmatrix}
\begin{tabular}{cccc|cccc}
$\alpha_1$ & $\alpha_x$ & $\alpha_{x^2}$ & $\alpha_{x^3}$ & $\alpha_{y}$ & $\alpha_{xy}$ & $\alpha_{x^2y}$ & $\alpha_{x^3y}$ \\
$\alpha_{x}$                        & $\alpha_{1}$                        & $\alpha_{x^3}$                        & $\alpha_{x^2}$                        & $\alpha_{xy}$                        & $\alpha_{y}$                        & $\alpha_{x^3y}$                        & $\alpha_{x^2y}$                        \\
$\alpha_{x^2}$                        & $\alpha_{x^3}$                        & $\alpha_{1}$                        & $\alpha_{x}$                        & $\alpha_{x^2y}$                        & $\alpha_{x^3y}$                        & $\alpha_{y}$                        & $\alpha_{xy}$                        \\
$\alpha_{x^3}$                        & $\alpha_{x^2}$                        & $\alpha_{x}$                        & $\alpha_{1}$                        & $\alpha_{x^3y}$                        & $\alpha_{x^2y}$                        & $\alpha_{xy}$                        & $\alpha_{y}$                        \\ \hline
$\alpha_{y}$ & $\alpha_{x^3y}$ & $\alpha_{x^2y}$ & $\alpha_{xy}$ & $\alpha_{1}$ & $\alpha_{x^3}$ & $\alpha_{x^2}$ & $\alpha_{x}$ \\
$\alpha_{xy}$                        & $\alpha_{y}$                        & $\alpha_{x^3y}$                        & $\alpha_{x^2y}$                        & $\alpha_{x}$                        & $\alpha_{1}$                        & $\alpha_{x^3}$                        & $\alpha_{x^2}$                                               \\
$\alpha_{x^2y}$                        & $\alpha_{xy}$                        & $\alpha_{y}$                        & $\alpha_{x^3y}$                        & $\alpha_{x^2}$                        & $\alpha_{x}$                        & $\alpha_{1}$                        & $\alpha_{x^3}$                                               \\
$\alpha_{x^3y}$                        & $\alpha_{x^2y}$                        & $\alpha_{xy}$                        & $\alpha_{y}$                        & $\alpha_{x^3}$                        & $\alpha_{x^2}$                        & $\alpha_{x}$                        & $\alpha_{1}$                                              
\end{tabular}
\end{pmatrix}.$$
We now look at the rows of $\Omega(v)$ and see what their corresponding element in $RD_8$ is, in terms of $v.$ Let $r_1, r_2,\dots , r_8$ be the rows of $\Omega(v),$ then each row is formed by multiplying each term of $v$ by an element of $G.$ The elements of $G$   do not have to be the same but they can be. For example:
$$r_1=\alpha_{(1)1}(1)1+\alpha_{(1)x}(1)x+\alpha_{(1)x^2}(1)x^2+\alpha_{(1)x^3}(1)x^3+\alpha_{(1)y}(1)y+$$
$$+\alpha_{(1)xy}(1)xy+\alpha_{(1)x^2y}(1)x^2y+\alpha_{(1)x^3y}(1)x^3y,$$
the first row of $\Omega(v)$ is obtained by multiplying each term of $v$ by the same group element of $G,$ namely $1.$
$$r_2=\alpha_{(x)1}(x)1+\alpha_{(x^3)x}(x^3)x+\alpha_{(x)x^2}(x)x^2+\alpha_{(x^3)x^3}(x^3)x^3+\alpha_{(x)y}(x)y+$$
$$+\alpha_{(x^3)xy}(x^3)xy+\alpha_{(x)x^2y}(x)x^2y+\alpha_{(x^3)x^3y}(x^3)x^3y,$$
the second row of $\Omega(v)$ is obtained by multiplying the terms of $v$ by the group elements of $G;$ $x$ or $x^3.$ 
$$r_8=\alpha_{(x^3y)1}(x^3y)1+\alpha_{(x^3y)x}(x^3y)x+\alpha_{(x^3y)x^2}(x^3y)x^2+\alpha_{(x^3y)x^3}(x^3y)x^3+\alpha_{(x^3y)y}(x^3y)y+$$
$$+\alpha_{(x^3y)xy}(x^3y)xy+\alpha_{(x^3y)x^2y}(x^3y)x^2y+\alpha_{(x^3y)x^3y}(x^3y)x^3y,$$
the eighth row of $\Omega(v)$ is obtained by multiplying each term of $v$ by the same group element of $G,$ namely $x^3y.$
\end{exam}

Example 1 highlights the difference between the matrix $\sigma(v)$ from \cite{XX} and the matrix $\Omega(v).$ Namely, each row of $\sigma(v)$ consists of vectors that correspond to the elements $hv$ in $RG$ with $h \in G$ (we multiply each term of $v$ by the same group element of $G$) where in $\Omega(v),$ some rows are formed by multiplying the terms of $v$ by different group elements of $G.$ Therefore, we can say that each row of $\Omega(v)$ corresponds to an element in $RG$ of the following form: 

\begin{equation}\label{ElementofOmega}
v_j^*=\sum_{i=1}^n \alpha_{g_{j_i}g_i}g_{j_i}g_i,
\end{equation}
where $\alpha_{g_{j_i}g_i} \in R,$ $g_i, g_{j_i} \in G$ and $j$ is the $jth$ row of the matrix $\Omega(v).$ In other words, we can define the composite matrix $\Omega(v)$ as:

\begin{equation}
\Omega(v)=\begin{pmatrix}
\alpha_{g_{1_1}g_1} & \alpha_{g_{1_2}g_2} & \alpha_{g_{1_3}g_3} & \dots & \alpha_{g_{1_n}g_n}\\
\alpha_{g_{2_1}g_1} & \alpha_{g_{2_2}g_2} & \alpha_{g_{2_3}g_3} & \dots & \alpha_{g_{2_n}g_n}\\
\vdots & \vdots & \vdots & \vdots & \vdots \\
\alpha_{g_{n_1}g_1} & \alpha_{g_{n_2}g_2} & \alpha_{g_{n_3}g_3} & \dots & \alpha_{g_{n_n}g_n}\\
\end{pmatrix},
\end{equation}
where the elements $g_{j_i}$ are simply the group elements $G.$ Which elements of $G$ these are, depends how the composite matrix is defined, i.e., what groups we employ and how we define the $\phi_l$ map in individual blocks.

It is possible to form a composite matrix so that each row of $\Omega(v)$ corresponds to the elements $\alpha_{g_{j_i}}v$ in $RG$ where $g_{j_i}$ are equal for all $i \in \{1,2,\dots,n\}.$ If this is the case, then $\Omega(v)$ is equivalent to $\sigma(v).$ We look at an example.

\begin{exam}\label{ex.3}
Let $G=\langle x,y \ | \ x^4=y^2=1, x^y=x^{-1} \rangle \cong D_8.$ Let $v=\alpha_1+\alpha_{x}x+\alpha_{x^2}x^2+\alpha_{x^3}x^3+\alpha_{y}y+\alpha_{xy}xy+\alpha_{x^2y}x^2y+\alpha_{x^3y}x^3y \in RD_8,$ where $\alpha_{g_i} \in R.$ Then

$$\sigma(v)=\begin{pmatrix}
\begin{tabular}{cccc|cccc}
$\alpha_1$ & $\alpha_x$ & $\alpha_{x^2}$ & $\alpha_{x^3}$ & $\alpha_{y}$ & $\alpha_{xy}$ & $\alpha_{x^2y}$ & $\alpha_{x^3y}$ \\
$\alpha_{x^3}$                        & $\alpha_{1}$                        & $\alpha_{x}$                        & $\alpha_{x^2}$                        & $\alpha_{x^3y}$                        & $\alpha_{y}$                        & $\alpha_{xy}$                        & $\alpha_{x^2y}$                        \\
$\alpha_{x^2}$                        & $\alpha_{x^3}$                        & $\alpha_{1}$                        & $\alpha_{x}$                        & $\alpha_{x^2y}$                        & $\alpha_{x^3y}$                        & $\alpha_{y}$                        & $\alpha_{xy}$                        \\
$\alpha_{x}$                        & $\alpha_{x^2}$                        & $\alpha_{x^3}$                        & $\alpha_{1}$                        & $\alpha_{xy}$                        & $\alpha_{x^2y}$                        & $\alpha_{x^3y}$                        & $\alpha_{y}$                        \\ \hline
$\alpha_{y}$ & $\alpha_{x^3y}$ & $\alpha_{x^2y}$ & $\alpha_{xy}$ & $\alpha_{1}$ & $\alpha_{x^3}$ & $\alpha_{x^2}$ & $\alpha_{x}$ \\
$\alpha_{xy}$                        & $\alpha_{y}$                        & $\alpha_{x^3y}$                        & $\alpha_{x^2y}$                        & $\alpha_{x}$                        & $\alpha_{1}$                        & $\alpha_{x^3}$                        & $\alpha_{x^2}$                                               \\
$\alpha_{x^2y}$                        & $\alpha_{xy}$                        & $\alpha_{y}$                        & $\alpha_{x^3y}$                        & $\alpha_{x^2}$                        & $\alpha_{x}$                        & $\alpha_{1}$                        & $\alpha_{x^3}$                                               \\
$\alpha_{x^3y}$                        & $\alpha_{x^2y}$                        & $\alpha_{xy}$                        & $\alpha_{y}$                        & $\alpha_{x^3}$                        & $\alpha_{x^2}$                        & $\alpha_{x}$                        & $\alpha_{1}$                                              
\end{tabular}
\end{pmatrix}.$$

Now let $H_1=\langle a \ | \ a^4=1 \rangle \cong C_4$ and define the composite matrix as:
$$\Omega(v)=\begin{pmatrix}
A_1'&A_2'\\
A_3&A_4
\end{pmatrix}=$$

\resizebox{0.7\textwidth}{!}{\begin{minipage}{\textwidth}
$$\begin{pmatrix}
\begin{tabular}{cccc|cccc}
$\alpha_{{g_1^{-1}g_1}}$ & $\alpha_{{g_1^{-1}g_2}}$ & $\alpha_{{g_1^{-1}g_3}}$ & $\alpha_{{g_1^{-1}g_4}}$ & $\alpha_{{g_1^{-1}g_5}}$ & $\alpha_{{g_1^{-1}g_6}}$ & $\alpha_{{g_1^{-1}g_7}}$ & $\alpha_{{g_1^{-1}g_8}}$ \\
$\alpha_{\phi_1((h_1)_2^{-1}(h_1)_1)}$                        & $\alpha_{\phi_1((h_1)_2^{-1}(h_1)_2)}$                        & $\alpha_{\phi_1((h_1)_2^{-1}(h_1)_3)}$                        & $\alpha_{\phi_1((h_1)_2^{-1}(h_1)_4)}$                        & $\alpha_{\phi_2((h_2)_2^{-1}(h_2)_1)}$                        & $\alpha_{\phi_2((h_2)_2^{-1}(h_2)_2)}$                        & $\alpha_{\phi_2((h_2)_2^{-1}(h_2)_3)}$                        & $\alpha_{\phi_2((h_2)_2^{-1}(h_2)_4)}$                        \\
$\alpha_{\phi_1((h_1)_3^{-1}(h_1)_1)}$                        & $\alpha_{\phi_1((h_1)_3^{-1}(h_1)_2)}$                        & $\alpha_{\phi_1(h_1)_3^{-1}(h_1)_3)}$                        & $\alpha_{\phi_1((h_1)_3^{-1}(h_1)_4)}$                        & $\alpha_{\phi_2((h_2)_3^{-1}(h_2)_1)}$                        & $\alpha_{\phi_2((h_2)_3^{-1}(h_2)_2)}$                        & $\alpha_{\phi_2((h_2)_3^{-1}(h_2)_3)}$                        & $\alpha_{\phi_2((h_2)_3^{-1}(h_2)_4)}$                        \\
$\alpha_{\phi_1((h_1)_4^{-1}(h_1)_1)}$                        & $\alpha_{\phi_1((h_1)_4^{-1}(h_1)_2)}$                        & $\alpha_{\phi_1(h_1)_4^{-1}(h_1)_3)}$                        & $\alpha_{\phi_1((h_1)_4^{-1}(h_1)_4)}$                        & $\alpha_{\phi_2((h_2)_4^{-1}(h_2)_1)}$                        & $\alpha_{\phi_2((h_2)_4^{-1}(h_2)_2)}$                        & $\alpha_{\phi_2((h_2)_4^{-1}(h_2)_3)}$                        & $\alpha_{\phi_2((h_2)_4^{-1}(h_2)_4)}$                        \\ \hline
$\alpha_{{g_5^{-1}g_1}}$ & $\alpha_{{g_5^{-1}g_2}}$ & $\alpha_{{g_5^{-1}g_3}}$ & $\alpha_{{g_5^{-1}g_4}}$ & $\alpha_{{g_5^{-1}g_5}}$ & $\alpha_{{g_5^{-1}g_6}}$ & $\alpha_{{g_5^{-1}g_7}}$ & $\alpha_{{g_5^{-1}g_8}}$ \\
$\alpha_{{g_6^{-1}g_1}}$                        & $\alpha_{{g_6^{-1}g_2}}$                        & $\alpha_{{g_6^{-1}g_3}}$                        & $\alpha_{{g_6^{-1}g_4}}$                        & $\alpha_{{g_6^{-1}g_5}}$                        & $\alpha_{{g_6^{-1}g_6}}$                        & $\alpha_{{g_6^{-1}g_7}}$                        & $\alpha_{{g_6^{-1}g_8}}$                                               \\
$\alpha_{{g_7^{-1}g_1}}$                        & $\alpha_{{g_7^{-1}g_2}}$                        & $\alpha_{{g_7^{-1}g_3}}$                        & $\alpha_{{g_7^{-1}g_4}}$                        & $\alpha_{{g_7^{-1}g_5}}$                        & $\alpha_{{g_7^{-1}g_6}}$                        & $\alpha_{{g_7^{-1}g_7}}$                        & $\alpha_{{g_7^{-1}g_8}}$                                               \\
$\alpha_{{g_8^{-1}g_1}}$                        & $\alpha_{{g_8^{-1}g_2}}$                        & $\alpha_{{g_8^{-1}g_3}}$                        & $\alpha_{{g_8^{-1}g_4}}$                        & $\alpha_{{g_8^{-1}g_5}}$                        & $\alpha_{{g_8^{-1}g_6}}$                        & $\alpha_{{g_8^{-1}g_7}}$                        & $\alpha_{{g_8^{-1}g_8}}$                                              
\end{tabular}
\end{pmatrix},$$
\end{minipage}}\\

where:
\begin{table}[h!]
\centering
\begin{tabular}{cclcc}
\multirow{2}{*}{$\phi_1:$} & $(h_1)_i \xrightarrow{\phi_1} g_1^{-1}g_i$ &  & \multirow{2}{*}{$\phi_2:$} & $(h_1)_i \xrightarrow{\phi_2} g_1^{-1}g_j$    \\
                           & $\text{for} \ i=\{1,2,3,4\}$               &  &                            & $\text{for when} \ \{i=1,j=5,i=2,j=6,i=3,j=7,i=4,j=8\}$ 
\end{tabular}
\end{table}

in $A_1'$ and $A_2'.$ Then

$$\Omega(v)=\begin{pmatrix}
\begin{tabular}{cccc|cccc}
$\alpha_1$ & $\alpha_x$ & $\alpha_{x^2}$ & $\alpha_{x^3}$ & $\alpha_{y}$ & $\alpha_{xy}$ & $\alpha_{x^2y}$ & $\alpha_{x^3y}$ \\
$\alpha_{x^3}$                        & $\alpha_{1}$                        & $\alpha_{x}$                        & $\alpha_{x^2}$                        & $\alpha_{x^3y}$                        & $\alpha_{y}$                        & $\alpha_{xy}$                        & $\alpha_{x^2y}$                        \\
$\alpha_{x^2}$                        & $\alpha_{x^3}$                        & $\alpha_{1}$                        & $\alpha_{x}$                        & $\alpha_{x^2y}$                        & $\alpha_{x^3y}$                        & $\alpha_{y}$                        & $\alpha_{xy}$                        \\
$\alpha_{x}$                        & $\alpha_{x^2}$                        & $\alpha_{x^3}$                        & $\alpha_{1}$                        & $\alpha_{xy}$                        & $\alpha_{x^2y}$                        & $\alpha_{x^3y}$                        & $\alpha_{y}$                        \\ \hline
$\alpha_{y}$ & $\alpha_{x^3y}$ & $\alpha_{x^2y}$ & $\alpha_{xy}$ & $\alpha_{1}$ & $\alpha_{x^3}$ & $\alpha_{x^2}$ & $\alpha_{x}$ \\
$\alpha_{xy}$                        & $\alpha_{y}$                        & $\alpha_{x^3y}$                        & $\alpha_{x^2y}$                        & $\alpha_{x}$                        & $\alpha_{1}$                        & $\alpha_{x^3}$                        & $\alpha_{x^2}$                                               \\
$\alpha_{x^2y}$                        & $\alpha_{xy}$                        & $\alpha_{y}$                        & $\alpha_{x^3y}$                        & $\alpha_{x^2}$                        & $\alpha_{x}$                        & $\alpha_{1}$                        & $\alpha_{x^3}$                                               \\
$\alpha_{x^3y}$                        & $\alpha_{x^2y}$                        & $\alpha_{xy}$                        & $\alpha_{y}$                        & $\alpha_{x^3}$                        & $\alpha_{x^2}$                        & $\alpha_{x}$                        & $\alpha_{1}$                                              
\end{tabular}
\end{pmatrix}.$$

Clearly, in this specific case,  $\Omega(v)$ is equivalent to $\sigma(v).$
\end{exam}

Example \ref{ex.3} leads to the following result.

\begin{cor}
The matrix $\Omega(v)$ is equivalent to the matrix $\sigma(v)$ if the group elements $g_{j_i}$ in Equation \ref{ElementofOmega} are the same for all $i \in \{1,2,\dots,n\}.$
\begin{proof}
If $g_{j_1}=g_{j_2}=g_{j_3}=\dots =g_{j_n}$ in Equation \ref{ElementofOmega}, then each row of $\Omega(v)$ corresponds to an element $g_jv$ in $RG$ where $g_j$ is any element of $G.$ This is exactly what each row of $\sigma(v)$ corresponds to in $RG.$ Thus $\Omega(v)$ is equivalent to $\sigma(v).$
\end{proof}
\end{cor}

\begin{cor}
Let $v=\alpha_{g_1}g_1+\alpha_{g_2}g_2+\dots \alpha_{g_n}g_n \in RG$ and $\sigma(v)$ be the corresponding matrix over $R.$ Let $v'$ be also an element of $RG$ but with a different ordering of the elements to $v$. Then $\sigma(v')$ is permutation equivalent to $\sigma(v).$
\begin{proof}
Without loss of generality, let $v'=\alpha_{g_3}g_3+\alpha_{g_2}g_2+\alpha_{g_1}g_1+\alpha_{g_n}g_n+\alpha_{g_{n-1}}g_{n-1}+\dots +\alpha_{g_{n-3}}g_{n-3}.$ Then
$$\sigma(v')=\begin{pmatrix}
\alpha_{g_3^{-1}g_3}&\alpha_{g_3^{-1}g_2}&\alpha_{g_3^{-1}g_1}&\alpha_{g_3^{-1}g_n}&\alpha_{g_3^{-1}g_{n-1}}&\dots &\alpha_{g_3^{-1}g_{n-3}}\\
\alpha_{g_2^{-1}g_3}&\alpha_{g_2^{-1}g_2}&\alpha_{g_2^{-1}g_1}&\alpha_{g_2^{-1}g_n}&\alpha_{g_2^{-1}g_{n-1}}&\dots &\alpha_{g_2^{-1}g_{n-3}}\\
\alpha_{g_1^{-1}g_3}&\alpha_{g_1^{-1}g_2}&\alpha_{g_1^{-1}g_1}&\alpha_{g_1^{-1}g_n}&\alpha_{g_1^{-1}g_{n-1}}&\dots &\alpha_{g_1^{-1}g_{n-3}}\\
\alpha_{g_n^{-1}g_3}&\alpha_{g_n^{-1}g_2}&\alpha_{g_n^{-1}g_1}&\alpha_{g_n^{-1}g_n}&\alpha_{g_n^{-1}g_{n-1}}&\dots &\alpha_{g_n^{-1}g_{n-3}}\\
\alpha_{g_{n-1}^{-1}g_3}&\alpha_{g_{n-1}^{-1}g_2}&\alpha_{g_{n-1}^{-1}g_1}&\alpha_{g_{n-1}^{-1}g_n}&\alpha_{g_{n-1}^{-1}g_{n-1}}&\dots &\alpha_{g_{n-1}^{-1}g_{n-3}}\\
\vdots&\vdots&\vdots&\vdots&\vdots&\vdots&\vdots\\
\alpha_{g_{n-3}^{-1}g_3}&\alpha_{g_{n-3}^{-1}g_2}&\alpha_{g_{n-3}^{-1}g_1}&\alpha_{g_{n-3}^{-1}g_n}&\alpha_{g_{n-3}^{-1}g_{n-1}}&\dots &\alpha_{g_{n-3}^{-1}g_{n-3}}\\
\end{pmatrix}.$$
It is clear that $\sigma(v')$ is row and column permutation equivalent to $\sigma(v).$ This concludes the proof.
\end{proof}
\end{cor}

We now show when the matrix $\Omega(v)$ is not permutation equivalent to $\sigma(v)$ for any arrangement of the elements of $G.$

\begin{thm}
Let $\Omega(v)$ be a composite matrix over $R$ such that at least one row of $\Omega(v)$ corresponds to an element in $RG$ of the form 

$$v_j^*=\sum_{i=1}^n \alpha_{g_{j_i}g_i}g_{j_i}g_i,$$
where $g_{j_i}$ is not the same for all $i \in \{1,2,\dots ,n\}.$ Here, $\alpha_{g_{j_i}g_i} \in R,$ $g_i, g_{j_i} \in G$ and $j$ is the $jth$ row of the matrix $\Omega(v).$ Then $\Omega(v)$ is not permutation equivalent to $\sigma(v)$ for any arrangement of the elements of $G$ in $v.$

\begin{proof}
Assume that the matrix $\Omega(v)$ is permutation equivalent to the matrix $\sigma(v)$ for some arrangements of the elements of $G$ in $v.$ This is equivalent to say that the matrix $\sigma(v)$ is permutation equivalent to the matrix $\Omega(v)$ for some specific arrangement of the elements of $G$ in $v.$ But we know from the previous corollary that for any arrangement of the elements of $G$ in $v,$ the corresponding matrix is permutation equivalent to $\sigma(v).$ This implies that for any arrangement of the elements of $G$ in the group ring element $v,$ the rows of the corresponding matrix will be of the form $hv$ where $h$ is any element of $G.$ In other words, there will be no row that will correspond to an element in $RG$ of the form $v_j^*=\sum_{i=1}^n \alpha_{g_{j_i}g_i}g_{j_i}g_i,$ where $g_{j_i}$ is not the same for all $i \in \{1,2,\dots ,n\}.$ Thus, $\sigma(v)$ is not permutation equivalent to $\Omega(v)$ for any arrangements of the elements of $G$ in $v.$ This contradicts our assumption. Therefore, the matrix $\Omega(v)$ is not permutation equivalent to the matrix $\sigma(v)$ for any arrangement of the elements of $G$ in $v.$ 
\end{proof}

\end{thm}

\section{Composite $G$-codes}

We are now ready to introduce the code construction.

For a given element $v \in RG$ and some groups $H_i$ of order $r,$ we define the following code over the ring $R:$

\begin{equation}
\mathcal{C}(v)=\langle \Omega(v) \rangle.
\end{equation}

The code is formed by taking the row space of $\Omega(v)$ over the ring $R.$ As in \cite{XIX}, the code $\mathcal{C}(v)$ is a linear code over the ring $R$,  since it is the row space of a generator matrix. It is not possible to determine the size of the code immediately from the matrix. 

\begin{exam}\label{Ex.1}
Let $G=\langle x,y \ | \ x^4=1, y^2=x^2, x^y=x^{-1} \rangle \cong Q_8.$ Let $v=\sum_{i=0}^3 \alpha_{i+1} x^i+\alpha_{i+5} x^iy \in RQ_8,$ where $\alpha_i=\alpha_{g_i} \in R.$ Let $H_1=\langle a,b \ | \ a^2=b^2=1, ab=ba \rangle \cong C_2 \times C_2.$ We now define the composite matrix as:
$$\Omega(v)=\begin{pmatrix}
A_1'&A_2\\
A_3&A_4'
\end{pmatrix}=$$

\resizebox{0.7\textwidth}{!}{\begin{minipage}{\textwidth}
$$\begin{pmatrix}
\begin{tabular}{cccc|cccc}
$\alpha_{{g_1^{-1}g_1}}$ & $\alpha_{{g_1^{-1}g_2}}$ & $\alpha_{{g_1^{-1}g_3}}$ & $\alpha_{{g_1^{-1}g_4}}$ & $\alpha_{{g_1^{-1}g_5}}$ & $\alpha_{{g_1^{-1}g_6}}$ & $\alpha_{{g_1^{-1}g_7}}$ & $\alpha_{{g_1^{-1}g_8}}$ \\
$\alpha_{\phi_1((h_1)_2^{-1}(h_1)_1)}$                        & $\alpha_{\phi_1((h_1)_2^{-1}(h_1)_2)}$                        & $\alpha_{\phi_1((h_1)_2^{-1}(h_1)_3)}$                        & $\alpha_{\phi_1((h_1)_2^{-1}(h_1)_4)}$                        & $\alpha_{{g_2^{-1}g_5}}$                        & $\alpha_{{g_2^{-1}g_6}}$                        & $\alpha_{{g_2^{-1}g_7}}$                        & $\alpha_{{g_2^{-1}g_8}}$                        \\
$\alpha_{\phi_1((h_1)_3^{-1}(h_1)_1)}$                        & $\alpha_{\phi_1((h_1)_3^{-1}(h_1)_2)}$                        & $\alpha_{\phi_1((h_1)_3^{-1}(h_1)_3)}$                        & $\alpha_{\phi_1((h_1)_3^{-1}(h_1)_4)}$                        & $\alpha_{{g_3^{-1}g_5}}$                        & $\alpha_{{g_3^{-1}g_6}}$                        & $\alpha_{{g_3^{-1}g_7}}$                        & $\alpha_{{g_3^{-1}g_8}}$                        \\
$\alpha_{\phi_1((h_1)_4^{-1}(h_1)_1)}$                        & $\alpha_{\phi_1((h_1)_4^{-1}(h_1)_2)}$                        & $\alpha_{\phi_1((h_1)_4^{-1}(h_1)_3)}$                        & $\alpha_{\phi_1((h_1)_4^{-1}(h_1)_4)}$                        & $\alpha_{{g_4^{-1}g_5}}$                        & $\alpha_{{g_4^{-1}g_6}}$                        & $\alpha_{{g_4^{-1}g_7}}$                        & $\alpha_{{g_4^{-1}g_8}}$                        \\ \hline
$\alpha_{{g_5^{-1}g_1}}$ & $\alpha_{{g_5^{-1}g_2}}$ & $\alpha_{{g_5^{-1}g_3}}$ & $\alpha_{{g_5^{-1}g_4}}$ & $\alpha_{{g_5^{-1}g_5}}$ & $\alpha_{{g_5^{-1}g_6}}$ & $\alpha_{{g_5^{-1}g_7}}$ & $\alpha_{{g_5^{-1}g_8}}$ \\
$\alpha_{{g_6^{-1}g_1}}$                        & $\alpha_{{g_6^{-1}g_2}}$                        & $\alpha_{{g_6^{-1}g_3}}$                        & $\alpha_{{g_6^{-1}g_4}}$                        & $\alpha_{\phi_4((h_1)_2^{-1}(h_1)_1)}$                        & $\alpha_{\phi_4((h_1)_2^{-1}(h_1)_2)}$                        & $\alpha_{\phi_4((h_1)_2^{-1}(h_1)_3)}$                        & $\alpha_{\phi_4((h_1)_2^{-1}(h_1)_4)}$                                               \\
$\alpha_{{g_7^{-1}g_1}}$                        & $\alpha_{{g_7^{-1}g_2}}$                        & $\alpha_{{g_7^{-1}g_3}}$                        & $\alpha_{{g_7^{-1}g_4}}$                        & $\alpha_{\phi_4((h_1)_3^{-1}(h_1)_1)}$                        & $\alpha_{\phi_4((h_1)_3^{-1}(h_1)_2)}$                        & $\alpha_{\phi_4((h_1)_3^{-1}(h_1)_3)}$                        & $\alpha_{\phi_4((h_1)_3^{-1}(h_1)_4)}$                                               \\
$\alpha_{{g_8^{-1}g_1}}$                        & $\alpha_{{g_8^{-1}g_2}}$                        & $\alpha_{{g_8^{-1}g_3}}$                        & $\alpha_{{g_8^{-1}g_4}}$                        & $\alpha_{\phi_4((h_1)_4^{-1}(h_1)_1)}$                        & $\alpha_{\phi_4((h_1)_4^{-1}(h_1)_2)}$                        & $\alpha_{\phi_4((h_1)_4^{-1}(h_1)_3)}$                        & $\alpha_{\phi_4((h_1)_4^{-1}(h_1)_4)}$                                              
\end{tabular}
\end{pmatrix},$$
\end{minipage}}\\

where:
\begin{table}[h!]
\centering
\begin{tabular}{cclcc}
\multirow{2}{*}{$\phi_1:$} & $(h_1)_i \xrightarrow{\phi_1} g_1^{-1}g_i$ &  & \multirow{2}{*}{$\phi_4:$} & $(h_1)_i \xrightarrow{\phi_4} g_5^{-1}g_j$    \\
                           & $\text{for} \ i=\{1,2,3,4\}$               &  &                            & $\text{for when} \ \{i=1,j=5,i=2,j=6,i=3,j=7,i=4,j=8\},$ 
\end{tabular}
\end{table}

\noindent in $A_1'$ and $A_4'$ respectively. This results in a composite matrix over $R$ of the following form:

$$\Omega(v)=\begin{pmatrix}
\begin{tabular}{cc|cc}
$X_1$                & $Y_1$                & \multicolumn{2}{c}{\multirow{2}{*}{$X_2$}} \\
$Y_1$                & $X_1$                & \multicolumn{2}{c}{}                       \\ \hline
\multicolumn{2}{c|}{\multirow{2}{*}{$X_3$}} & $X_4$                & $Y_4$               \\
\multicolumn{2}{c|}{}                       & $Y_4$                & $X_4$              
\end{tabular}
\end{pmatrix}=\begin{pmatrix}
\begin{tabular}{cccc|cccc}
$\alpha_1$ & $\alpha_2$ & $\alpha_3$ & $\alpha_4$ & $\alpha_5$ & $\alpha_6$ & $\alpha_7$ & $\alpha_8$ \\
$\alpha_2$                        & $\alpha_1$                        & $\alpha_4$                        & $\alpha_3$                        & $\alpha_8$                        & $\alpha_5$                        & $\alpha_6$                        & $\alpha_7$                        \\
$\alpha_3$                        & $\alpha_4$                        & $\alpha_1$                        & $\alpha_2$                        & $\alpha_7$                        & $\alpha_8$                        & $\alpha_5$                        & $\alpha_6$                        \\
$\alpha_4$                        & $\alpha_3$                        & $\alpha_2$                        & $\alpha_1$                        & $\alpha_6$                        & $\alpha_7$                        & $\alpha_8$                        & $\alpha_5$                        \\ \hline
$\alpha_7$ & $\alpha_6$ & $\alpha_5$ & $\alpha_8$ & $\alpha_1$ & $\alpha_4$ & $\alpha_3$ & $\alpha_2$ \\
$\alpha_8$                        & $\alpha_7$                        & $\alpha_6$                        & $\alpha_5$                        & $\alpha_4$                        & $\alpha_1$                        & $\alpha_2$                        & $\alpha_3$                                               \\
$\alpha_5$                        & $\alpha_8$                        & $\alpha_7$                        & $\alpha_6$                        & $\alpha_3$                        & $\alpha_2$                        & $\alpha_1$                        & $\alpha_4$                                               \\
$\alpha_6$                        & $\alpha_5$                        & $\alpha_8$                        & $\alpha_7$                        & $\alpha_2$                        & $\alpha_3$                        & $\alpha_4$                        & $\alpha_1$                                              
\end{tabular}
\end{pmatrix}.$$
If we  let $v=x^3+xy+x^2y+x^3y \in \mathbb{F}_2Q_8$, where $\langle x,y \rangle \cong Q_8,$ then 
$$\mathcal{C}(v)=\langle \Omega(v)\rangle =\begin{pmatrix}
0&0&0&1&0&1&1&1\\
0&0&1&0&1&0&1&1\\
0&1&0&0&1&1&0&1\\
1&0&0&0&1&1&1&0\\
1&1&0&1&0&1&0&0\\
1&1&1&0&1&0&0&0\\
0&1&1&1&0&0&0&1\\
1&0&1&1&0&0&1&0
\end{pmatrix}$$
and $\mathcal{C}(v)$ is equivalent to 
$$\begin{pmatrix}
1&0&0&0&0&1&1&1\\
0&1&0&0&1&0&1&1\\
0&0&1&0&1&1&0&1\\
0&0&0&1&1&1&1&0
\end{pmatrix}.$$
Clearly $\mathcal{C}(v)=\langle \Omega(v) \rangle$ is the $[8,4,4]$ extended Hamming code.
\end{exam}

In the above example, the group $C_2 \times C_2$ was applied twice in two different blocks: $A_1'$ and $A_4'.$ As mentioned in the previous section, we can employ more than one group of order $r$. We look at another example.

\begin{exam}
Let $G=\langle x,y \ | \ x^4=1, y^2=x^2, x^y=x^{-1} \rangle \cong Q_8.$ Let $v=\sum_{i=0}^3 \alpha_{i+1} x^i+\alpha_{i+5} x^iy \in RQ_8,$ where $\alpha_i=\alpha_{g_i} \in R.$ Let $H_1=\langle a,b \ | \ a^2=b^2=1, ab=ba \rangle \cong C_2 \times C_2$ and $H_2=\langle c \ | \ c^4=1 \rangle \cong C_4.$ We now define the composite matrix as:
$$\Omega(v)=\begin{pmatrix}
A_1'&A_2\\
A_3&A_4'
\end{pmatrix}=$$

\resizebox{0.7\textwidth}{!}{\begin{minipage}{\textwidth}
$$\begin{pmatrix}
\begin{tabular}{cccc|cccc}
$\alpha_{{g_1^{-1}g_1}}$ & $\alpha_{{g_1^{-1}g_2}}$ & $\alpha_{{g_1^{-1}g_3}}$ & $\alpha_{{g_1^{-1}g_4}}$ & $\alpha_{{g_1^{-1}g_5}}$ & $\alpha_{{g_1^{-1}g_6}}$ & $\alpha_{{g_1^{-1}g_7}}$ & $\alpha_{{g_1^{-1}g_8}}$ \\
$\alpha_{\phi_1((h_1)_2^{-1}(h_1)_1)}$                        & $\alpha_{\phi_1((h_1)_2^{-1}(h_1)_2)}$                        & $\alpha_{\phi_1((h_1)_2^{-1}(h_1)_3)}$                        & $\alpha_{\phi_1((h_1)_2^{-1}(h_1)_4)}$                        & $\alpha_{{g_2^{-1}g_5}}$                        & $\alpha_{{g_2^{-1}g_6}}$                        & $\alpha_{{g_2^{-1}g_7}}$                        & $\alpha_{{g_2^{-1}g_8}}$                        \\
$\alpha_{\phi_1((h_1)_3^{-1}(h_1)_1)}$                        & $\alpha_{\phi_1((h_1)_3^{-1}(h_1)_2)}$                        & $\alpha_{\phi_1((h_1)_3^{-1}(h_1)_3)}$                        & $\alpha_{\phi_1((h_1)_3^{-1}(h_1)_4)}$                        & $\alpha_{{g_3^{-1}g_5}}$                        & $\alpha_{{g_3^{-1}g_6}}$                        & $\alpha_{{g_3^{-1}g_7}}$                        & $\alpha_{{g_3^{-1}g_8}}$                        \\
$\alpha_{\phi_1((h_1)_4^{-1}(h_1)_1)}$                        & $\alpha_{\phi_1((h_1)_4^{-1}(h_1)_2)}$                        & $\alpha_{\phi_1((h_1)_4^{-1}(h_1)_3)}$                        & $\alpha_{\phi_1((h_1)_4^{-1}(h_1)_4)}$                        & $\alpha_{{g_4^{-1}g_5}}$                        & $\alpha_{{g_4^{-1}g_6}}$                        & $\alpha_{{g_4^{-1}g_7}}$                        & $\alpha_{{g_4^{-1}g_8}}$                        \\ \hline
$\alpha_{{g_5^{-1}g_1}}$ & $\alpha_{{g_5^{-1}g_2}}$ & $\alpha_{{g_5^{-1}g_3}}$ & $\alpha_{{g_5^{-1}g_4}}$ & $\alpha_{{g_5^{-1}g_5}}$ & $\alpha_{{g_5^{-1}g_6}}$ & $\alpha_{{g_5^{-1}g_7}}$ & $\alpha_{{g_5^{-1}g_8}}$ \\
$\alpha_{{g_6^{-1}g_1}}$                        & $\alpha_{{g_6^{-1}g_2}}$                        & $\alpha_{{g_6^{-1}g_3}}$                        & $\alpha_{{g_6^{-1}g_4}}$                        & $\alpha_{\phi_4((h_2)_2^{-1}(h_2)_1)}$                        & $\alpha_{\phi_4((h_2)_2^{-1}(h_2)_2)}$                        & $\alpha_{\phi_4((h_2)_2^{-1}(h_2)_3)}$                        & $\alpha_{\phi_4((h_2)_2^{-1}(h_2)_4)}$                                               \\
$\alpha_{{g_7^{-1}g_1}}$                        & $\alpha_{{g_7^{-1}g_2}}$                        & $\alpha_{{g_7^{-1}g_3}}$                        & $\alpha_{{g_7^{-1}g_4}}$                        & $\alpha_{\phi_4((h_2)_3^{-1}(h_2)_1)}$                        & $\alpha_{\phi_4((h_2)_3^{-1}(h_2)_2)}$                        & $\alpha_{\phi_4((h_2)_3^{-1}(h_2)_3)}$                        & $\alpha_{\phi_4((h_2)_3^{-1}(h_2)_4)}$                                               \\
$\alpha_{{g_8^{-1}g_1}}$                        & $\alpha_{{g_8^{-1}g_2}}$                        & $\alpha_{{g_8^{-1}g_3}}$                        & $\alpha_{{g_8^{-1}g_4}}$                        & $\alpha_{\phi_4((h_2)_4^{-1}(h_2)_1)}$                        & $\alpha_{\phi_4((h_2)_4^{-1}(h_2)_2)}$                        & $\alpha_{\phi_4((h_2)_4^{-1}(h_2)_3)}$                        & $\alpha_{\phi_4((h_2)_4^{-1}(h_2)_4)}$                                              
\end{tabular}
\end{pmatrix},$$
\end{minipage}}\\

where:
\begin{table}[h!]
\centering
\begin{tabular}{cclcc}
\multirow{2}{*}{$\phi_1:$} & $(h_1)_i \xrightarrow{\phi_1} g_1^{-1}g_i$ &  & \multirow{2}{*}{$\phi_4:$} & $(h_2)_i \xrightarrow{\phi_4} g_5^{-1}g_j$    \\
                           & $\text{for} \ i=\{1,2,3,4\}$               &  &                            & $\text{for when} \ \{i=1,j=5,i=2,j=6,i=3,j=7,i=4,j=8\},$ 
\end{tabular}
\end{table}

\noindent in $A_1'$ and $A_4'$ respectively. This results in a composite matrix over $R$ of the following form:

$$\Omega(v)=\begin{pmatrix}
\begin{tabular}{cc|cc}
$X_1$                & $Y_1$                & \multicolumn{2}{c}{\multirow{2}{*}{$X_2$}} \\
$Y_1$                & $X_1$                & \multicolumn{2}{c}{}                       \\ \hline
\multicolumn{2}{c|}{\multirow{2}{*}{$X_3$}} & $X_4$                & $Y_4$               \\
\multicolumn{2}{c|}{}                       & $Y_4'$                & $X_4$              
\end{tabular}
\end{pmatrix}=\begin{pmatrix}
\begin{tabular}{cccc|cccc}
$\alpha_1$ & $\alpha_2$ & $\alpha_3$ & $\alpha_4$ & $\alpha_5$ & $\alpha_6$ & $\alpha_7$ & $\alpha_8$ \\
$\alpha_2$                        & $\alpha_1$                        & $\alpha_4$                        & $\alpha_3$                        & $\alpha_8$                        & $\alpha_5$                        & $\alpha_6$                        & $\alpha_7$                        \\
$\alpha_3$                        & $\alpha_4$                        & $\alpha_1$                        & $\alpha_2$                        & $\alpha_7$                        & $\alpha_8$                        & $\alpha_5$                        & $\alpha_6$                        \\
$\alpha_4$                        & $\alpha_3$                        & $\alpha_2$                        & $\alpha_1$                        & $\alpha_6$                        & $\alpha_7$                        & $\alpha_8$                        & $\alpha_5$                        \\ \hline
$\alpha_7$ & $\alpha_6$ & $\alpha_5$ & $\alpha_8$ & $\alpha_1$ & $\alpha_4$ & $\alpha_3$ & $\alpha_2$ \\
$\alpha_8$                        & $\alpha_7$                        & $\alpha_6$                        & $\alpha_5$                        & $\alpha_4$                        & $\alpha_1$                        & $\alpha_2$                        & $\alpha_3$                                               \\
$\alpha_5$                        & $\alpha_8$                        & $\alpha_7$                        & $\alpha_6$                        & $\alpha_2$                        & $\alpha_3$                        & $\alpha_1$                        & $\alpha_4$                                               \\
$\alpha_6$                        & $\alpha_5$                        & $\alpha_8$                        & $\alpha_7$                        & $\alpha_3$                        & $\alpha_2$                        & $\alpha_4$                        & $\alpha_1$                                              
\end{tabular}
\end{pmatrix}.$$
If we  let $v=x^3+xy+x^2y+x^3y \in \mathbb{F}_2Q_8$, where $\langle x,y \rangle \cong Q_8,$ then 
$$\mathcal{C}(v)=\langle \Omega(v)\rangle =\begin{pmatrix}
0&0&0&1&0&1&1&1\\
0&0&1&0&1&0&1&1\\
0&1&0&0&1&1&0&1\\
1&0&0&0&1&1&1&0\\
1&1&0&1&0&1&0&0\\
1&1&1&0&1&0&0&0\\
0&1&1&1&0&0&0&1\\
1&0&1&1&0&0&1&0
\end{pmatrix}$$
and $\mathcal{C}(v)$ is equivalent to 
$$\begin{pmatrix}
1&0&0&0&0&1&1&1\\
0&1&0&0&1&0&1&1\\
0&0&1&0&1&1&0&1\\
0&0&0&1&1&1&1&0
\end{pmatrix}.$$
Clearly $\mathcal{C}(v)=\langle \Omega(v) \rangle$ is the $[8,4,4]$ extended Hamming code.
\end{exam}

We now extend two results from \cite{XIX}; we show that the codes constructed from the composite matrices are also ideals in the group ring. We then show that the automorphism group of such codes contains the group $G$ as a subgroup.

\begin{thm}
Let $R$ be a finite commutative Frobenius ring, $G$ a finite group of order $n.$ Let $H_i$ be finite groups of order $r$ such that $r$ is a factor of $n$ with $n>r$ and $n,r \neq 1.$ Also, let $v \in RG$ and $\mathcal{C}(v)=\langle \Omega(v) \rangle$ be the corresponding code in $R^n.$ Define $I(v)$ to be the set of elements of $RG$ such that $\sum \alpha_ig_i \in I(v)$ if and only if $(\alpha_1,\alpha_2,\dots,\alpha_n) \in \mathcal{C}(v).$ Then $I(v)$ is a left ideal in $RG.$
\begin{proof}
We saw above that the rows of $\Omega(v)$ consist precisely of the vectors that correspond to the elements of the form $v_j^*=\sum_{i=1}^n \alpha_{g_{j_i}g_i}g_{j_i}g_i$ in $RG,$ where $\alpha_{g_{j_i}g_i} \in R,$ $g_i, g_{j_i} \in G$ and $j$ is the $jth$ row of the matrix $\Omega(v).$ We also know that some of the elements $g_{j_i}$ equal to $\phi_l(h_i)$ for some map $\phi_l$ and the elements $h_i$ of $H_i.$ Let $a=\sum \alpha_ig_i$ and $b=\sum \beta_ig_i$ be two elements in $I(v),$ then $a+b=\sum (\alpha_i+\beta_i)g_i$ which corresponds to the sum of the corresponding elements in $\mathcal{C}(v).$ This implies that $I(v)$ is closed under addition.

Let $w_1=\sum \beta_i g_i \in RG.$ Then if $w_2$ corresponds to a vector in $\mathcal{C}(v),$ it is of the form $\sum \gamma_j v_j^*.$ Then $w_1w_2=\sum \beta_i g_i \sum \gamma_j v_j^*=\sum \beta_i \gamma_jg_i v_j^*$ which corresponds to an element in $\mathcal{C}(v)$ and gives that the element is in $I(v).$ Therefore $I(v)$ is a left ideal of $RG.$
\end{proof}
\end{thm}

\begin{cor}\label{Cor.3.2}
Let $R$ be a finite commutative Frobenius ring and $G$ a finite group of order $n.$ Let $H_i$ be finite groups of order $r$ such that $r$ is a factor of $n$ with $n>r$ and $n,r \neq 1.$ Also, let $v \in RG$ and let $\mathcal{C}(v)=\langle \Omega(v) \rangle$ be the corresponding code in $R^n.$ Then the automorphism group of $\mathcal{C}(v)$ has a subgroup isomorphic to the group $G.$
\begin{proof}
Since $I(v)$ is an ideal in $RG$ we have that $I(v)$ is held invariant by the action of the elements of the group $G.$ It follows immediately that the automorphism group of $\mathcal{C}(v)$ contains the group $G$  as a subgroup.
\end{proof}
\end{cor}

Similarly, as in \cite{XIX}, the codes constructed by the above technique are held invariant by the action of the group $G$ on the coordinates. We can therefore construct a code whose automorphism group must contain the group $G.$ Moreover, in our construction, we apply groups of order $r$ and the bijective maps $\phi_l$ in individual blocks to determine the permutation of the coordinates in each row of a code. For this reason, we refer to a code constructed by the above technique as a composite $G$-code. 

We also have the following as a result of Corollary~\ref{Cor.3.2}.

\begin{cor}
The putative $[72,36,16]$ code cannot be of the form $\mathcal{C}(v)= \langle \Omega(v) \rangle$ for any $v \in \mathbb{F}_2G$ for any group $G.$
\end{cor}
\begin{proof}
It is well known that the automorphism group of a putative $[72,36,16]$ code must have order less than or equal to $5$ (see \cite{XIX} for details).  If it were of this construction, some group of order $72$ would have to be in its automorphism group.  Therefore, the code cannot be formed from this construction.
\end{proof} 

We finish this section with one more result which is a generalization  of the result from $\cite{XIX}.$ We show that if $\mathcal{C}$ is a composite $G$-code for some $G$ then its orthogonal $\mathcal{C}^{\bot}$ is also a composite $G$-code.

Let $I$ be an ideal in a group ring $RG.$ Define $\mathcal{R}(\mathcal{C})=\{w \ | \ vw=0, \ \forall v \in I\}.$ It is immediate that $\mathcal{R}(I)$ is an ideal of $RG.$

Let $v=\alpha_{g_1}g_1+\alpha_{g_2}g_2+\dots +\alpha_{g_n}g_n \in RG$ and $\mathcal{C}(v)$ be the corresponding code. Let $\Psi: RG \rightarrow R^n$ be the canonical map that sends $\alpha_{g_1}g_1+\alpha_{g_2}g_2+\dots +\alpha_{g_n}g_n$ to $(\alpha_{g_1},\alpha_{g_2},\dots ,\alpha_{g_n}).$ Let $I$ be the ideal $\Psi^{-1}(\mathcal{C}).$ Let $\textbf{w}=(w_1,w_2,\dots,w_n) \in \mathcal{C}^{\bot}.$ Then
\begin{equation}
[(\alpha_{g_{j_i}g_1},\alpha_{g_{j_i}g_2}, \dots ,\alpha_{g_{j_i}g_n}),(w_1,w_2,\dots ,w_n)]=0, \ \forall j,
\end{equation}
where $g_{j_i} \in G.$ This gives that
\begin{equation}
\sum_{i=0}^n \alpha_{g_{j_i}g_i}w_i=0, \ \forall j.
\end{equation}

Let $w=\Psi^{-1}(\textbf{w})=\sum w_{g_i}g_i$ and define $\overline{\textbf{w}} \in RG$ to be $\overline{\textbf{w}}=b_{g_1}g_1+b_{g_2}g_2+\dots +b_{g_n}g_n$ where

\begin{equation}
b_{g_i}=w_{g_i^{-1}}.
\end{equation}

Then

\begin{equation}
\sum_{i=1}^n \alpha_{g_{j_i}g_i}w_i=0 \ \implies \ \sum_{i=1}^n \alpha_{g_{j_i}g_i}b_{g_i^{-1}}=0.
\end{equation}

Here $g_{j_i}g_ig_i^{-1}=g_{j_i},$ hence this is the the coefficient of $g_{j_i}$ in the product of $\overline{\textbf{w}}$ and $v_j^*.$ This gives that $\overline{\textbf{w}} \in \mathcal{R}(I)$ if and only if $\textbf{w} \in \mathcal{C}^{\bot}.$

Let $\phi: R^n \rightarrow RG$ by $\phi(\textbf{w})=\overline{\textbf{w}}.$ It is clear that $\phi$ is a bijection between $\mathcal{C}^{\bot}$ and $\mathcal{R}(\Psi^{-1}(\mathcal{C})).$

\begin{thm}
Let $\mathcal{C}=\mathcal{C}(v)$ be a code in $RG$ formed from the vector $v \in RG.$ Then $\Psi^{-1}(\mathcal{C}^{\bot})$ is an ideal of $RG.$

\begin{proof}
We have that $\Psi(\phi(\mathcal{C}^{\bot}))$ is permutation equivalent to $\mathcal{C}^{\bot}$ and $\phi(\mathcal{C}^{\bot})$ is an ideal and so $\Psi^{-1}(\mathcal{C})$ is an ideal as well.
\end{proof}
\end{thm}

\section{Self-Orthogonal Composite $G$-codes}

In this section, we extend more results from \cite{XIX}. Namely, we show that the map $\Omega : RG \rightarrow M_n(R)$ is an injective ring homomorphism, we show when our construction $\mathcal{C}=\langle \Omega(v) \rangle$ produces a self-orthogonal code and also when it produces a self-dual code.

Before we look at the theoretical results, we define the composite matrix $\Omega(v)$ that we defined in the the previous section, in a different but equivalent form. Namely, let 

$$\Omega(v)=\begin{pmatrix}
\alpha_{g_{1_1}^{-1}g_1} & \alpha_{g_{1_2}^{-1}g_2} & \alpha_{g_{1_3}^{-1}g_3} & \dots & \alpha_{g_{1_n}^{-1}g_n}\\
\alpha_{g_{2_1}^{-1}g_1} & \alpha_{g_{2_2}^{-1}g_2} & \alpha_{g_{2_3}^{-1}g_3} & \dots & \alpha_{g_{2_n}^{-1}g_n}\\
\vdots & \vdots & \vdots & \vdots & \vdots \\
\alpha_{g_{n_1}^{-1}g_1} & \alpha_{g_{n_2}^{-1}g_2} & \alpha_{g_{n_3}^{-1}g_3} & \dots & \alpha_{g_{n_n}^{-1}g_n}\\
\end{pmatrix},$$
where $g_{j_i}^{-1}$ are simply the elements of the group $G.$ These elements are determined by how the matrix has been partitioned, what groups $H_i$ of order $r$ have been employed and how the maps $\phi_l$ have been defined to form the composite matrix. This representation of the composite matrix $\Omega(v)$ will make it easier to prove the upcoming results.

\begin{thm}
Let $R$ be a finite commutative Frobenius ring, $G$ be a group of order $n$ and $H_i$ be finite groups of order $r$ such that $r$ is a factor of $n$ with $n>1$ and $n,r \neq 1.$ Then the map $\Omega : RG \rightarrow M_n(R)$ is an injective ring homomorphism.

\begin{proof}
We need to show that the map $\Omega$ preserves addition and multiplication. Let $R$ be a finite commutative Frobenius ring, $G$ be a group of order $n$ and $H_i$ be finite groups of order $r$ such that $r$ is a factor of $n$ with $n>1$ and $n,r \neq 1.$ Now define the mapping $\Omega : RG \rightarrow M_n(R)$ as follows. Suppose $v=\sum_{i=1}^n \alpha_{g_i} g_i.$ Then 

$$\Omega(v)=\begin{pmatrix}
\alpha_{g_{1_1}^{-1}g_1} & \alpha_{g_{1_2}^{-1}g_2} & \alpha_{g_{1_3}^{-1}g_3} & \dots & \alpha_{g_{1_n}^{-1}g_n}\\
\alpha_{g_{2_1}^{-1}g_1} & \alpha_{g_{2_2}^{-1}g_2} & \alpha_{g_{2_3}^{-1}g_3} & \dots & \alpha_{g_{2_n}^{-1}g_n}\\
\vdots & \vdots & \vdots & \vdots & \vdots \\
\alpha_{g_{n_1}^{-1}g_1} & \alpha_{g_{n_2}^{-1}g_2} & \alpha_{g_{n_3}^{-1}g_3} & \dots & \alpha_{g_{n_n}^{-1}g_n}\\
\end{pmatrix},$$
where $g_{j_i}^{-1}$ are simply the elements of the group $G$ in some order. This order is determined by how the matrix has been partitioned, what groups $H_i$ of order $r$ have been employed and how the maps $\phi_l$ have been defined to form the composite matrix $\Omega(v).$ This mapping is clearly surjective and injective. We now show that $\Omega$ is additive and multiplicative. Let $w=\sum_{i=1}^n \beta_{g_i}g_i$ then,

$$\Omega(v+w)=\begin{pmatrix}
(\alpha+\beta)_{g_{1_1}^{-1}g_1} & (\alpha+\beta)_{g_{1_2}^{-1}g_2} & (\alpha+\beta)_{g_{1_3}^{-1}g_3} & \dots & (\alpha+\beta)_{g_{1_n}^{-1}g_n}\\
(\alpha+\beta)_{g_{2_1}^{-1}g_1} & (\alpha+\beta)_{g_{2_2}^{-1}g_2} & (\alpha+\beta)_{g_{2_3}^{-1}g_3} & \dots & (\alpha+\beta)_{g_{2_n}^{-1}g_n}\\
\vdots & \vdots & \vdots & \vdots & \vdots \\
(\alpha+\beta)_{g_{n_1}^{-1}g_1} & (\alpha+\beta)_{g_{n_2}^{-1}g_2} & (\alpha+\beta)_{g_{n_3}^{-1}g_3} & \dots & (\alpha+\beta)_{g_{n_n}^{-1}g_n}\\
\end{pmatrix}=$$

$$=\begin{pmatrix}
\alpha_{g_{1_1}^{-1}g_1} & \alpha_{g_{1_2}^{-1}g_2} & \alpha_{g_{1_3}^{-1}g_3} & \dots & \alpha_{g_{1_n}^{-1}g_n}\\
\alpha_{g_{2_1}^{-1}g_1} & \alpha_{g_{2_2}^{-1}g_2} & \alpha_{g_{2_3}^{-1}g_3} & \dots & \alpha_{g_{2_n}^{-1}g_n}\\
\vdots & \vdots & \vdots & \vdots & \vdots \\
\alpha_{g_{n_1}^{-1}g_1} & \alpha_{g_{n_2}^{-1}g_2} & \alpha_{g_{n_3}^{-1}g_3} & \dots & \alpha_{g_{n_n}^{-1}g_n}\\
\end{pmatrix}+\begin{pmatrix}
\beta_{g_{1_1}^{-1}g_1} & \beta_{g_{1_2}^{-1}g_2} & \beta_{g_{1_3}^{-1}g_3} & \dots & \beta_{g_{1_n}^{-1}g_n}\\
\beta_{g_{2_1}^{-1}g_1} & \beta_{g_{2_2}^{-1}g_2} & \beta_{g_{2_3}^{-1}g_3} & \dots & \alpha_{g_{2_n}^{-1}g_n}\\
\vdots & \vdots & \vdots & \vdots & \vdots \\
\beta_{g_{n_1}^{-1}g_1} & \beta_{g_{n_2}^{-1}g_2} & \beta_{g_{n_3}^{-1}g_3} & \dots & \beta_{g_{n_n}^{-1}g_n}\\
\end{pmatrix}=$$

$$=\Omega(v)+\Omega(w).$$

Thus addition is preserved. Next, suppose $v \ast w=t,$ where $t=\sum_{i=1}^n \gamma_{g_i}g_i.$ Then

$$\Omega(v) \ast \Omega(w)=\begin{pmatrix}
\gamma_{g_{1_1}^{-1}g_1} & \gamma_{g_{1_2}^{-1}g_2} & \gamma_{g_{1_3}^{-1}g_3} & \dots & \gamma_{g_{1_n}^{-1}g_n}\\
\gamma_{g_{2_1}^{-1}g_1} & \gamma_{g_{2_2}^{-1}g_2} & \gamma_{g_{2_3}^{-1}g_3} & \dots & \gamma_{g_{2_n}^{-1}g_n}\\
\vdots & \vdots & \vdots & \vdots & \vdots \\
\gamma_{g_{n_1}^{-1}g_1} & \gamma_{g_{n_2}^{-1}g_2} & \gamma_{g_{n_3}^{-1}g_3} & \dots & \gamma_{g_{n_n}^{-1}g_n}\\
\end{pmatrix}=\Omega(v \ast w).$$
Thus, multiplication is preserved. This concludes the proof. 
\end{proof}

\end{thm}

For an element $v=\sum \alpha_ig_i \in RG,$ define the element $v^T \in RG$ as $v^T=\sum \alpha_ig_i^{-1}.$ This is sometimes known as the canonical involution for the group ring.

\begin{lem}
Let $R$ be a finite commutative Frobenius ring, $G$ be a group of order $n$ and $H_i$ be finite groups of order $r$ such that $r$ is a factor of $n$ with $n>1$ and $n,r \neq 1.$ Then for an element $v \in RG,$ we have that $\Omega(v)^T=\Omega(v^T).$

\begin{proof}
The $ij$-th elements of $\Omega(v^T)$ is $\alpha_{(g_i^{-1}g_{j_i})^{-1}}=\alpha_{g_{j_i}^{-1}g_i}$ which is the $ji$-th element of $\Omega(v).$
\end{proof}
\end{lem}

\begin{lem}
Let $R$ be a finite commutative Frobenius ring, $G$ be a group of order $n$ and $H_i$ be finite groups of order $r$ such that $r$ is a factor of $n$ with $n>1$ and $n,r \neq 1.$ If $v=v^T$ and $v^2=0$ then $\mathcal{C}_v$ is a self-orthogonal code.
\begin{proof}
If $v=v^T$ then $\Omega(v)^T=\Omega(v^T)$ by Lemma 4.2. Then we have that $(\Omega(v)\Omega(v))_{ij}$ is the inner-product of the i-th and j-th rows of $\Omega(v).$ Since $v^2=0,$ by Theorem 4.1 we have that $\Omega(v)\Omega(v)=\textbf{0}.$ This gives that any two rows of $\Omega(v)$ are orthogonal and hence they generate a self-orthogonal code.
\end{proof}
\end{lem}

\begin{thm}
Let $R$ be a finite commutative Frobenius ring, $G$ be a group of order $n$ and $H_i$ be finite groups of order $r$ such that $r$ is a factor of $n$ with $n>1$ and $n,r \neq 1.$ Let $v$ be an element in $RG.$ If $v=v^T, v^2=0,$ and $|\mathcal{C}_v|=|R^{\frac{n}{2}}|$ then $C_v$ is a self-dual code.
\end{thm}
\begin{proof}
By Lemma 4.3 the code $C_v$ is self-orthogonal and since $|C_v|=|R^{\frac{n}{2}}|$, we have that $\mathcal{C}_v$ is self-dual.
\end{proof}

\section{Quasi Composite $G$-codes}

In this section, we make a generalization of the notion of quasi-$G$-codes. In \cite{XXIII}, the authors have developed a ring with a Gray map that could be used to describe certain families of quasi-cyclic groups. That same ring can be used in this setting to construct quasi-composite $G$-codes which we shall describe below. Self-dual codes over these rings were studied in \cite{XXIV}. Recently, in \cite{XXV}, the authors study the algebraic structure of quasi-group codes.

Let $G$ be a finite group of order $n$ and $R$ a finite Frobenius commutative ring. Let $\mathcal{C}$ be a code in $R^{sn}$ where the coordinates can be partitioned into $n$ sets of size $s$ where each set is assigned an element of $G.$ If the code $\mathcal{C}$ is held invariant by the action of multiplying the coordinate set marker by every element of $G$ then the code $\mathcal{C}$ is called a quasi-composite $G$-code of index $s.$

We now describe a family of rings to construct quasi-composite $G$-codes.

Let $p_1,p_2,\dots ,p_t$ be prime numbers with $t\geq 0$ and $p_i \neq p_j$ if $i \neq j.$ Define $\Delta$ to be $\Delta=p_1^{k_1}p_2^{k_2}\dots p_t^{k_t},$ for some $k_i \geq 1, i=1,\dots ,t.$

The ring is defined as follows:
$$R_{q,\Delta}=\mathbb{F}_q[u_{p_1,1},\dots ,u_{p_1,k_1},u_{p_2,1}, \dots ,u_{p_2,k_2},\dots ,u_{p_t,k_t}]/\langle u_{p_i,j}^{p_i}=0\rangle ,$$ 

where the indeterminates $\{u_{p_i,j}\}_{(1\leq i \leq t, 1\leq j \leq k_i)}$ commute.

Let $i \in \{1,\dots ,t\}, j \in \{1,\dots , k_i\}.$ Take the set of exponents $J_i=\{0,1,\dots ,p_i-1\}$ for the indeterminant $u_{p_i,j}.$ For $\alpha_i \in J_i^{k_i}$ denote $u_{p_i,1}^{\alpha_i,1} \dots u_{p_i,k_i}^{\alpha_i,k_i}$ by $u_i^{\alpha_i}.$ For a monomial $u_1^{\alpha_1} \dots u_t^{\alpha_t}$ in $R_{q,\Delta}$ write $u^{\alpha},$ where $\alpha=(\alpha_1,\dots ,\alpha_t) \in J_1^{k_1} \times \dots \times J_t^{k_t}.$

Let $J=J_1^{k_1} \times \dots \times J_t^{k_t}.$ Any element $c$ in $R_{q,\Delta}$ can be written as

\begin{equation}
c=\sum_{\alpha \in J} c_{\alpha}u^{\alpha}=\sum_{\alpha \in J} c_{\alpha}u_{p_1,1}^{\alpha_1,1}\dots u_{p_1,k_1}^{\alpha_1,k_1}\dots u_{p_t,1}^{\alpha_t,1}\dots u_{p_t,k_t}^{\alpha_t,k_t},
\end{equation}
with $c_{\alpha} \in \mathbb{F}_q.$

It is immediate that $R_{q,\Delta}$ is a commutative ring with $|R_{q,\Delta}|=q^{p_1^{k_1}p_2^{k_2}\dots p_t^{k_t}}=q^{\Delta}.$

Next we define a Gray map on this ring. We will consider the elements in $R_{q,\Delta}$ as $q$-ary vectors of $\Delta$ coordinates. Order the elements of $A_{\Delta}$ lexicographically and use this ordering to label the coordinate positions of $\mathbb{F}_q^{\Delta}.$ Define the Gray map $\Psi_{\Delta} : A_{\Delta} \rightarrow \mathbb{F}_q^{\Delta}$ as follows:

\[ \Psi_{\Delta}(a)_b=   \left\{
\begin{array}{ll}
      1 & \text{if} \ \widehat{b} \subseteq \{\widehat{a} \cup 1\},  \\
      0 & \text{otherwise,} \\
\end{array} 
\right. \]
where $\Psi_{\Delta}(a)_b$ indicates the coordinate of $\Psi_{\Delta}(a)$ corresponding to the position of the element $b \in A_{\Delta}$ with the defined ordering.

It follows that $\Psi_{\Delta}(a)_b$ is 1, if each indeterminate $u_{p_1,j}$ in the monomial $b$ with non-zero exponent is also in the monomial $a$ with the same exponent. In other words, it is 1 when $\widehat{b}$ is a subset of $\widehat{a}.$ In order to consider all the subsets of $\widehat{a},$ we also add the empty subset that is given when $b=1;$ that is we compare $\widehat{b}$ to $\widehat{a} \cup 1.$

Finally, we extend $\Psi_{\Delta}$ linearly for all elements of $R_{q,\Delta}.$ Then $\Psi_{\Delta}$ is a Gray map from $R_{q,\Delta}$ to $\mathbb{F}_q^{\Delta}.$

\begin{thm}
Let $\mathcal{C}$ be a composite $G$-code in $R_{q,\Delta}$ for a finite group $G$ of order $n.$ Then $\Psi(\mathcal{C})$ is a quasi-composite $G$-code of length $n\Delta$ of index $\Delta$ in $\mathbb{F}_q^{\Delta n}.$

\begin{proof}
Since $\mathcal{C}$ is a composite $G$-code in $R_{q,\Delta},$ each row of $\mathcal{C}$ corresponds to an element of the form $v_j=c_{g_{j_1}g_1}g_{j_1}g_1+c_{g_{j_2}g_2}g_{j_2}g_2+\dots +c_{g_{j_n}g_n}g_{j_n}g_n$ in  $R_{q,\Delta}G,$ where $c_{g_{j_i}g_i} \in R_{q,\Delta}, g_{j_i}g_i \in G$ and where $j$ is the $jth$ row of the code $\mathcal{C}.$ Then $\Psi_{\Delta}(v_j)=\Psi_{\Delta}(c_{g_{j_1}g_1})g_{j_1}g_1+\Psi_{\Delta}(c_{g_{j_2}g_1})g_{j_2}g_2+\dots +\Psi_{\Delta}(c_{g_{j_n}g_1})g_{j_n}g_n.$ Therefore $\Psi_{\Delta}(\mathcal{C})$ is a quasi-composite $G$-code of length $n\Delta$ of index $\Delta$ in $\mathbb{F}_q^{\Delta n}.$
\end{proof}
\end{thm}

\begin{thm}
Let $\mathcal{C}$ be a composite $G$-code of length $n$ and of index $k$ over $R_{q,\Delta}$ for a finite group $G.$ Then $\Psi_{\Delta}(\mathcal{C})$ is a quasi-composite $G$-code of length $n\Delta$ of index $k\Delta$ in $\mathbb{F}_q^{\Delta n}.$

\begin{proof}
Since $\mathcal{C}$ is a composite $G$-code in $R_{q,\Delta},$ each row of $\mathcal{C}$ corresponds to an element of the form $v_j=c_{g_{j_1}g_1}g_{j_1}g_1+c_{g_{j_2}g_2}g_{j_2}g_2+\dots +c_{g_{j_n}g_n}g_{j_n}g_n$ in  $R_{q,\Delta}G,$ where $c_{g_{j_i}g_i} \in R_{q,\Delta}, g_{j_i}g_i \in G$ and where $j$ is the $jth$ row of the code $\mathcal{C}.$ Then $\Psi_{k\Delta}(v_j)=\Psi_{k\Delta}(c_{g_{j_1}g_1})g_{j_1}g_1+\Psi_{k\Delta}(c_{g_{j_2}g_1})g_{j_2}g_2+\dots +\Psi_{k\Delta}(c_{g_{j_n}g_1})g_{j_n}g_n.$ Therefore $\Psi_{\Delta}(\mathcal{C})$ is a quasi-composite $G$-code of length $n\Delta$ of index $k\Delta$ in $\mathbb{F}_q^{\Delta n}.$
\end{proof}
\end{thm}

We now look at constructing examples of quasi-composite $G$-codes. To do this, we first describe another family of rings which is useful in producing binary self-dual codes via their associated Gray maps.

Define the ring $R_k$ as

\begin{equation}
R_k=\mathbb{F}_2[u_1,u_2,\dots ,u_k]/\langle u_1^2uy_iu_j-u_ju_i \rangle.
\end{equation}

These rings are local rings of characteristic 2 with maximal ideal $\mathfrak{m}=\langle u_1,u_2,\dots,u_k \rangle.$ This maximal ideal is also necessarily the Jacobson radical of the ring, which can be characterized as the intersection of all maximal ideals. The socle, which is the sum of all minimal ideals, for the ring $R_k$ is $Soc(R_k)=\langle u_1u_2\dots u_k \rangle=\mathfrak{m}^{\bot}.$ We have that $|R_k|=2^{2^k}.$ The rings $R_k$ were described in \cite{II}, \cite{III}, and \cite{IV}.

We can describe a Gray map for $R_k.$ We define $\Phi_1(a+bu_1)=(b,a+b),$ where $\phi$ maps $R$ to $\mathbb{F}_2^2.$ Then view $R[u_1,u_2,\dots ,u_s]$ as $R[u_1,u_2,\dots ,u_{s-1}][u_s]$ and define $\phi_s(a+bu)=(b,a+b).$ Then the map $\phi_k$ is map from $R_k$ to $\mathbb{F}_2^{2^k}.$

The following theorem can be found in \cite{IV}.

\begin{thm}
Let $\mathcal{C}$ be a self-dual code over $R_k,$ then $\phi_k(\mathcal{C})$ is a self-dual code in $\mathbb{F}_2^{2^k}.$
\end{thm}

We note that the ring $R_k$ is $R_{2,2^k}$ in the setting of a family of rings to construct quasi-composite $G$-codes.

\begin{exam}
Let $G=\langle x ,y \ | \ x^8=y^2=1,x^y=x^{-1} \rangle \cong D_{16}$ and $v=\sum_{i=0}^7 \sum_{j=0}^1 \alpha_{1+i+8j}x^iy^j \in R_1D_{16}.$ Also let $H_1=\langle a \ | \ a^8=1 \rangle \cong C_8'.$ Next, we define the composite matrix as:

$$\Omega(v)=\begin{pmatrix}
A_1'&A_2'\\
A_3'&A_4'
\end{pmatrix},$$
where

\resizebox{0.65\textwidth}{!}{\begin{minipage}{\textwidth}
$$A_1'=\begin{pmatrix}
g_1^{-1}g_1&g_1^{-1}g_{2}&g_1^{-1}g_{3}&g_1^{-1}g_{4}&g_1^{-1}g_{5}&g_1^{-1}g_{6}&g_1^{-1}g_{7}&g_1^{-1}g_{8}\\
\phi_1((h_1)_2^{-1}(h_1)_1)&\phi_1((h_1)_2^{-1}(h_1)_2)&\phi_1((h_1)_2^{-1}(h_1)_3)&\phi_1((h_1)_2^{-1}(h_1)_4)&\phi_1((h_1)_2^{-1}(h_1)_5)&\phi_1((h_1)_2^{-1}(h_1)_6)&\phi_1((h_1)_2^{-1}(h_1)_7)&\phi_1((h_1)_2^{-1}(h_1)_8)\\
\phi_1((h_1)_3^{-1}(h_1)_1)&\phi_1((h_1)_3^{-1}(h_1)_2)&\phi_1((h_1)_3^{-1}(h_1)_3)&\phi_1((h_1)_3^{-1}(h_1)_4)&\phi_1((h_1)_3^{-1}(h_1)_5)&\phi_1((h_1)_3^{-1}(h_1)_6)&\phi_1((h_1)_3^{-1}(h_1)_7)&\phi_1((h_1)_3^{-1}(h_1)_8)\\
\phi_1((h_1)_4^{-1}(h_1)_1)&\phi_1((h_1)_4^{-1}(h_1)_2)&\phi_1((h_1)_4^{-1}(h_1)_3)&\phi_1((h_1)_4^{-1}(h_1)_4)&\phi_1((h_1)_4^{-1}(h_1)_5)&\phi_1((h_1)_4^{-1}(h_1)_6)&\phi_1((h_1)_4^{-1}(h_1)_7)&\phi_1((h_1)_4^{-1}(h_1)_8)\\
\phi_1((h_1)_5^{-1}(h_1)_1)&\phi_1((h_1)_5^{-1}(h_1)_2)&\phi_1((h_1)_5^{-1}(h_1)_3)&\phi_1((h_1)_5^{-1}(h_1)_4)&\phi_1((h_1)_5^{-1}(h_1)_5)&\phi_1((h_1)_5^{-1}(h_1)_6)&\phi_1((h_1)_5^{-1}(h_1)_7)&\phi_1((h_1)_5^{-1}(h_1)_8)\\
\phi_1((h_1)_6^{-1}(h_1)_1)&\phi_1((h_1)_6^{-1}(h_1)_2)&\phi_1((h_1)_6^{-1}(h_1)_3)&\phi_1((h_1)_6^{-1}(h_1)_4)&\phi_1((h_1)_6^{-1}(h_1)_5)&\phi_1((h_1)_6^{-1}(h_1)_6)&\phi_1((h_1)_6^{-1}(h_1)_7)&\phi_1((h_1)_6^{-1}(h_1)_8)\\
\phi_1((h_1)_7^{-1}(h_1)_1)&\phi_1((h_1)_7^{-1}(h_1)_2)&\phi_1((h_1)_7^{-1}(h_1)_3)&\phi_1((h_1)_7^{-1}(h_1)_4)&\phi_1((h_1)_7^{-1}(h_1)_5)&\phi_1((h_1)_7^{-1}(h_1)_6)&\phi_1((h_1)_7^{-1}(h_1)_7)&\phi_1((h_1)_7^{-1}(h_1)_8)\\
\phi_1((h_1)_8^{-1}(h_1)_1)&\phi_1((h_1)_8^{-1}(h_1)_2)&\phi_1((h_1)_8^{-1}(h_1)_3)&\phi_1((h_1)_8^{-1}(h_1)_4)&\phi_1((h_1)_8^{-1}(h_1)_5)&\phi_1((h_1)_8^{-1}(h_1)_6)&\phi_1((h_1)_8^{-1}(h_1)_7)&\phi_1((h_1)_8^{-1}(h_1)_8)\\
\end{pmatrix},$$
\end{minipage}}

\resizebox{0.65\textwidth}{!}{\begin{minipage}{\textwidth}
$$A_2'=\begin{pmatrix}
g_1^{-1}g_9&g_1^{-1}g_{10}&g_1^{-1}g_{11}&g_1^{-1}g_{12}&g_1^{-1}g_{13}&g_1^{-1}g_{14}&g_1^{-1}g_{15}&g_1^{-1}g_{16}\\
\phi_2((h_1)_2^{-1}(h_1)_1)&\phi_2((h_1)_2^{-1}(h_1)_2)&\phi_2((h_1)_2^{-1}(h_1)_3)&\phi_2((h_1)_2^{-1}(h_1)_4)&\phi_2((h_1)_2^{-1}(h_1)_5)&\phi_2((h_1)_2^{-1}(h_1)_6)&\phi_2((h_1)_2^{-1}(h_1)_7)&\phi_2((h_1)_2^{-1}(h_1)_8)\\
\phi_2((h_1)_3^{-1}(h_1)_1)&\phi_2((h_1)_3^{-1}(h_1)_2)&\phi_2((h_1)_3^{-1}(h_1)_3)&\phi_2((h_1)_3^{-1}(h_1)_4)&\phi_2((h_1)_3^{-1}(h_1)_5)&\phi_2((h_1)_3^{-1}(h_1)_6)&\phi_2((h_1)_3^{-1}(h_1)_7)&\phi_2((h_1)_3^{-1}(h_1)_8)\\
\phi_2((h_1)_4^{-1}(h_1)_1)&\phi_2((h_1)_4^{-1}(h_1)_2)&\phi_2((h_1)_4^{-1}(h_1)_3)&\phi_2((h_1)_4^{-1}(h_1)_4)&\phi_2((h_1)_4^{-1}(h_1)_5)&\phi_2((h_1)_4^{-1}(h_1)_6)&\phi_2((h_1)_4^{-1}(h_1)_7)&\phi_2((h_1)_4^{-1}(h_1)_8)\\
\phi_2((h_1)_5^{-1}(h_1)_1)&\phi_2((h_1)_5^{-1}(h_1)_2)&\phi_2((h_1)_5^{-1}(h_1)_3)&\phi_2((h_1)_5^{-1}(h_1)_4)&\phi_2((h_1)_5^{-1}(h_1)_5)&\phi_2((h_1)_5^{-1}(h_1)_6)&\phi_2((h_1)_5^{-1}(h_1)_7)&\phi_2((h_1)_5^{-1}(h_1)_8)\\
\phi_2((h_1)_6^{-1}(h_1)_1)&\phi_2((h_1)_6^{-1}(h_1)_2)&\phi_2((h_1)_6^{-1}(h_1)_3)&\phi_2((h_1)_6^{-1}(h_1)_4)&\phi_2((h_1)_6^{-1}(h_1)_5)&\phi_2((h_1)_6^{-1}(h_1)_6)&\phi_2((h_1)_6^{-1}(h_1)_7)&\phi_2((h_1)_6^{-1}(h_1)_8)\\
\phi_2((h_1)_7^{-1}(h_1)_1)&\phi_2((h_1)_7^{-1}(h_1)_2)&\phi_2((h_1)_7^{-1}(h_1)_3)&\phi_2((h_1)_7^{-1}(h_1)_4)&\phi_2((h_1)_7^{-1}(h_1)_5)&\phi_2((h_1)_7^{-1}(h_1)_6)&\phi_2((h_1)_7^{-1}(h_1)_7)&\phi_2((h_1)_7^{-1}(h_1)_8)\\
\phi_2((h_1)_8^{-1}(h_1)_1)&\phi_2((h_1)_8^{-1}(h_1)_2)&\phi_2((h_1)_8^{-1}(h_1)_3)&\phi_2((h_1)_8^{-1}(h_1)_4)&\phi_2((h_1)_8^{-1}(h_1)_5)&\phi_2((h_1)_8^{-1}(h_1)_6)&\phi_2((h_1)_8^{-1}(h_1)_7)&\phi_2((h_1)_8^{-1}(h_1)_8)\\
\end{pmatrix},$$
\end{minipage}}

\resizebox{0.65\textwidth}{!}{\begin{minipage}{\textwidth}
$$A_3'=\begin{pmatrix}
g_9^{-1}g_1&g_9^{-1}g_{2}&g_9^{-1}g_{3}&g_9^{-1}g_{4}&g_9^{-1}g_{5}&g_9^{-1}g_{6}&g_9^{-1}g_{7}&g_9^{-1}g_{8}\\
\phi_3((h_1)_2^{-1}(h_1)_1)&\phi_3((h_1)_2^{-1}(h_1)_2)&\phi_3((h_1)_2^{-1}(h_1)_3)&\phi_3((h_1)_2^{-1}(h_1)_4)&\phi_3((h_1)_2^{-1}(h_1)_5)&\phi_3((h_1)_2^{-1}(h_1)_6)&\phi_3((h_1)_2^{-1}(h_1)_7)&\phi_3((h_1)_2^{-1}(h_1)_8)\\
\phi_3((h_1)_3^{-1}(h_1)_1)&\phi_3((h_1)_3^{-1}(h_1)_2)&\phi_3((h_1)_3^{-1}(h_1)_3)&\phi_3((h_1)_3^{-1}(h_1)_4)&\phi_3((h_1)_3^{-1}(h_1)_5)&\phi_3((h_1)_3^{-1}(h_1)_6)&\phi_3((h_1)_3^{-1}(h_1)_7)&\phi_3((h_1)_3^{-1}(h_1)_8)\\
\phi_3((h_1)_4^{-1}(h_1)_1)&\phi_3((h_1)_4^{-1}(h_1)_2)&\phi_3((h_1)_4^{-1}(h_1)_3)&\phi_3((h_1)_4^{-1}(h_1)_4)&\phi_3((h_1)_4^{-1}(h_1)_5)&\phi_3((h_1)_4^{-1}(h_1)_6)&\phi_3((h_1)_4^{-1}(h_1)_7)&\phi_3((h_1)_4^{-1}(h_1)_8)\\
\phi_3((h_1)_5^{-1}(h_1)_1)&\phi_3((h_1)_5^{-1}(h_1)_2)&\phi_3((h_1)_5^{-1}(h_1)_3)&\phi_3((h_1)_5^{-1}(h_1)_4)&\phi_3((h_1)_5^{-1}(h_1)_5)&\phi_3((h_1)_5^{-1}(h_1)_6)&\phi_3((h_1)_5^{-1}(h_1)_7)&\phi_3((h_1)_5^{-1}(h_1)_8)\\
\phi_3((h_1)_6^{-1}(h_1)_1)&\phi_3((h_1)_6^{-1}(h_1)_2)&\phi_3((h_1)_6^{-1}(h_1)_3)&\phi_3((h_1)_6^{-1}(h_1)_4)&\phi_3((h_1)_6^{-1}(h_1)_5)&\phi_3((h_1)_6^{-1}(h_1)_6)&\phi_3((h_1)_6^{-1}(h_1)_7)&\phi_3((h_1)_6^{-1}(h_1)_8)\\
\phi_3((h_1)_7^{-1}(h_1)_1)&\phi_3((h_1)_7^{-1}(h_1)_2)&\phi_3((h_1)_7^{-1}(h_1)_3)&\phi_3((h_1)_7^{-1}(h_1)_4)&\phi_3((h_1)_7^{-1}(h_1)_5)&\phi_3((h_1)_7^{-1}(h_1)_6)&\phi_3((h_1)_7^{-1}(h_1)_7)&\phi_3((h_1)_7^{-1}(h_1)_8)\\
\phi_3((h_1)_8^{-1}(h_1)_1)&\phi_3((h_1)_8^{-1}(h_1)_2)&\phi_3((h_1)_8^{-1}(h_1)_3)&\phi_3((h_1)_8^{-1}(h_1)_4)&\phi_3((h_1)_8^{-1}(h_1)_5)&\phi_3((h_1)_8^{-1}(h_1)_6)&\phi_3((h_1)_8^{-1}(h_1)_7)&\phi_3((h_1)_8^{-1}(h_1)_8)\\
\end{pmatrix},$$
\end{minipage}}

\resizebox{0.65\textwidth}{!}{\begin{minipage}{\textwidth}
$$A_4'=\begin{pmatrix}
g_9^{-1}g_9&g_9^{-1}g_{10}&g_9^{-1}g_{11}&g_9^{-1}g_{12}&g_9^{-1}g_{13}&g_9^{-1}g_{14}&g_9^{-1}g_{15}&g_9^{-1}g_{16}\\
\phi_4((h_1)_2^{-1}(h_1)_1)&\phi_4((h_1)_2^{-1}(h_1)_2)&\phi_4((h_1)_2^{-1}(h_1)_3)&\phi_4((h_1)_2^{-1}(h_1)_4)&\phi_4((h_1)_2^{-1}(h_1)_5)&\phi_4((h_1)_2^{-1}(h_1)_6)&\phi_4((h_1)_2^{-1}(h_1)_7)&\phi_4((h_1)_2^{-1}(h_1)_8)\\
\phi_4((h_1)_3^{-1}(h_1)_1)&\phi_4((h_1)_3^{-1}(h_1)_2)&\phi_4((h_1)_3^{-1}(h_1)_3)&\phi_4((h_1)_3^{-1}(h_1)_4)&\phi_4((h_1)_3^{-1}(h_1)_5)&\phi_4((h_1)_3^{-1}(h_1)_6)&\phi_4((h_1)_3^{-1}(h_1)_7)&\phi_4((h_1)_3^{-1}(h_1)_8)\\
\phi_4((h_1)_4^{-1}(h_1)_1)&\phi_4((h_1)_4^{-1}(h_1)_2)&\phi_4((h_1)_4^{-1}(h_1)_3)&\phi_4((h_1)_4^{-1}(h_1)_4)&\phi_4((h_1)_4^{-1}(h_1)_5)&\phi_4((h_1)_4^{-1}(h_1)_6)&\phi_4((h_1)_4^{-1}(h_1)_7)&\phi_4((h_1)_4^{-1}(h_1)_8)\\
\phi_4((h_1)_5^{-1}(h_1)_1)&\phi_4((h_1)_5^{-1}(h_1)_2)&\phi_4((h_1)_5^{-1}(h_1)_3)&\phi_4((h_1)_5^{-1}(h_1)_4)&\phi_4((h_1)_5^{-1}(h_1)_5)&\phi_4((h_1)_5^{-1}(h_1)_6)&\phi_4((h_1)_5^{-1}(h_1)_7)&\phi_4((h_1)_5^{-1}(h_1)_8)\\
\phi_4((h_1)_6^{-1}(h_1)_1)&\phi_4((h_1)_6^{-1}(h_1)_2)&\phi_4((h_1)_6^{-1}(h_1)_3)&\phi_4((h_1)_6^{-1}(h_1)_4)&\phi_4((h_1)_6^{-1}(h_1)_5)&\phi_4((h_1)_6^{-1}(h_1)_6)&\phi_4((h_1)_6^{-1}(h_1)_7)&\phi_4((h_1)_6^{-1}(h_1)_8)\\
\phi_4((h_1)_7^{-1}(h_1)_1)&\phi_4((h_1)_7^{-1}(h_1)_2)&\phi_4((h_1)_7^{-1}(h_1)_3)&\phi_4((h_1)_7^{-1}(h_1)_4)&\phi_4((h_1)_7^{-1}(h_1)_5)&\phi_4((h_1)_7^{-1}(h_1)_6)&\phi_4((h_1)_7^{-1}(h_1)_7)&\phi_4((h_1)_7^{-1}(h_1)_8)\\
\phi_4((h_1)_8^{-1}(h_1)_1)&\phi_4((h_1)_8^{-1}(h_1)_2)&\phi_4((h_1)_8^{-1}(h_1)_3)&\phi_4((h_1)_8^{-1}(h_1)_4)&\phi_4((h_1)_8^{-1}(h_1)_5)&\phi_4((h_1)_8^{-1}(h_1)_6)&\phi_4((h_1)_8^{-1}(h_1)_7)&\phi_4((h_1)_8^{-1}(h_1)_8)\\
\end{pmatrix},$$
\end{minipage}}\\

and where:

\begin{table}[h!]
\centering
\begin{tabular}{cclcc}
\multirow{2}{*}{$\phi_1:$} & $(h_1)_i \xrightarrow{\phi_1} \alpha_{g_1^{-1}g_i}$ &  & \multirow{2}{*}{$\phi_2:$} & $(h_1)_i \xrightarrow{\phi_2} \alpha_{g_1^{-1}g_i}$    \\
                           & $\text{for} \ i=\{1,2,\dots,8\}$               &  &                            & $\text{for when} \ i=\{9,10,\dots,16\}$
\end{tabular}
\end{table}

\begin{table}[h!]
\centering
\begin{tabular}{cclcc}
\multirow{2}{*}{$\phi_3:$} & $(h_1)_i \xrightarrow{\phi_3} \alpha_{g_9^{-1}g_i}$ &  & \multirow{2}{*}{$\phi_4:$} & $(h_1)_i \xrightarrow{\phi_4} \alpha_{g_9^{-1}g_i}$    \\
                           & $\text{for} \ i=\{1,2,\dots,8\}$               &  &                            & $\text{for when} \ i=\{9,10,\dots,16\}.$
\end{tabular}
\end{table}

This results in a composite matrix of the following form:

$$\Omega(v)=\begin{pmatrix}
\begin{tabular}{cc|cc}
$A_1$ & $B_1$ & $A_2$ & $B_2$ \\
$B_1'$ & $A_1$ & $B_2'$ & $A_2$ \\ \hline
$A_3$ & $B_3$ & $A_4$ & $B_4$ \\
$B_3'$ & $A_3$ & $B_4'$ & $A_4$
\end{tabular}
\end{pmatrix},$$
where $A_1=circ(\alpha_1,\alpha_2,\alpha_3,\alpha_4), B_1=circ(\alpha_5,\alpha_6,\alpha_7,\alpha_8), B_1'=circ(\alpha_8,\alpha_5,\alpha_6,\alpha_7), A_2=circ(\alpha_9,\alpha_{10},\alpha_{11},\alpha_{12}), B_2=circ(\alpha_{13},\alpha_{14},\alpha_{15},\alpha_{16}), B_2'=circ(\alpha_{16},\alpha_{13},\alpha_{14},\alpha_{15}), A_3=circ(\alpha_9,$
$\alpha_{16},\alpha_{15},\alpha_{14}), B_3=circ(\alpha_{13},\alpha_{12},\alpha_{11},\alpha_{10}), B_3'=circ(\alpha_{10},\alpha_{13},\alpha_{12},\alpha_{11}), A_4=circ(\alpha_1,\alpha_8,\alpha_7,\alpha_6),$
$ B_4=circ(\alpha_5,\alpha_4,\alpha_3,\alpha_2), B_4'=circ(\alpha_2,\alpha_5,\alpha_4,\alpha_3)$ and where $\alpha_i \in R_1.$

If we  let $v=u(1+x+x^2+x^3+x^4+x^5+x^6)+1(x^7)+u(x^2y+x^3y)+1(x^5y+x^6y+x^7y) \in R_1D_{16}$ where $\langle x,y \rangle \cong D_{16},$ then

$$\Omega(v)=\begin{pmatrix}
\begin{tabular}{cc|cc}
$A_1$ & $B_1$ & $A_2$ & $B_2$ \\
$B_1'$ & $A_1$ & $B_2'$ & $A_2$ \\ \hline
$A_3$ & $B_3$ & $A_4$ & $B_4$ \\
$B_3'$ & $A_3$ & $B_4'$ & $A_4$
\end{tabular}
\end{pmatrix},$$
where $A_1=circ(u,u,u,u), B_1=circ(u,u,u,1), B_1'=circ(1,u,u,u), A_2=circ(0,0,u,u), B_2=circ(0,1,1,1), B_2'=circ(1,0,1,1), A_3=circ(0,1,1,1), B_3=circ(0,u,u,0), B_3'=circ(0,0,u,u),$
$A_4=circ(u,1,u,u), B_4=circ(u,u,u,u)$ and $B_4'=circ(u,u,u,u).$ Then, $\Omega(v)$ is a self-dual code of length $16$ over $R_1.$ Hence $\phi_1(\Omega(v))$ is a binary self-dual code of length 32 by Theorem~6.3. The binary code $\phi_1(\Omega(v))$ has a generator matrix of the following form:

$$\phi_1(\Omega(v))=\begin{pmatrix}
\begin{tabular}{cc|cc}
$A_1$ & $B_1$ & $A_2$ & $B_2$ \\
$B_1'$ & $A_1$ & $B_2'$ & $A_2$ \\ \hline
$A_3$ & $B_3$ & $A_4$ & $B_4$ \\
$B_3'$ & $A_3$ & $B_4'$ & $A_4$
\end{tabular}
\end{pmatrix},$$
where $A_1=circ(1,1,1,1,1,1,1,1), B_1=circ(1,1,1,1,1,1,0,1), B_1'=circ(0,1,1,1,1,1,1,1), A_2=circ(0,0,0,0,1,1,1,1), B_2=circ(0,0,0,1,0,1,0,1), B_2'=circ(0,1,0,0,0,1,0,1), A_3=circ(0,0,$
$0,1,0,1,0,1), B_3=circ(0,0,1,1,1,1,0,0), B_3'=circ(0,0,0,0,1,1,1,1),$
$A_4=circ(1,1,0,1,1,1,$
$1,1), B_4=circ(1,1,1,1,1,1,1,1)$ and $B_4'=circ(1,1,1,1,1,1,1,1).$

It is a simple calculation to see that $\phi_1(\Omega(v))$ is the extremal binary self-dual code with parameters $[32,16,8].$ It follows that this code is a quasi-composite $D_{16}$ code of index 2.
\end{exam}

\section{Generator matrices of the form $[I_n \ | \ \Omega(v)]$}

In this section, we consider generator matrices of the form $[I_n \ | \ \Omega(v)]$ to construct extremal binary self-dual codes. This approach was used in \cite{XXI} and \cite{XXII} where only groups of orders $4, 8$ and $16$ were considered to form the matrices $\Omega(v).$ In both papers: \cite{XXI} and \cite{XXII}, the authors define a specific generator matrices of the form $[I_n \ | \ \Omega(v)]$ for lengths $8$ and $16.$ The authors also prove theoretical results on when these matrices produce self-dual codes over the Frobenius ring $R.$ We generalize the theoretical results so that we show when the generator matrices of the form $[I_n \ | \ \Omega(v)]$ produce self-dual codes for any possible case rather than looking at individual cases for specific composite matrices $\Omega(v).$ Before the theoretical results, we give a motivating example in which we compare the generator matrix of the form $[I_n \ | \ \sigma(v)]$ with a generator matrix of the form $[I_n \ | \ \Omega(v)].$

\begin{exam}
Let $G=\langle x ,y \ | \ x^8=y^2=1,x^y=x^{-1} \rangle \cong D_{16}.$ Also let $v=\sum_{i=0}^7 \sum_{j=0}^1 \alpha_{1+i+8j}x^iy^j \in \mathbb{F}_2D_{16},$ then 
$$\sigma(v)=\begin{pmatrix}
A&B\\
B^T&A^T
\end{pmatrix},$$
where $A=circ(\alpha_1,\alpha_2,\alpha_3,\alpha_4,\alpha_5,\alpha_6,\alpha_7,\alpha_8),$ $B=circ(\alpha_9,\alpha_{10},\alpha_{11},\alpha_{12},\alpha_{13},\alpha_{14},\alpha_{15},\alpha_{16})$ and $\alpha_i \in \mathbb{F}_2.$ We now employ the generator matrix of the form $[I_{16} \ | \ \sigma(v)]$ to search for binary self-dual codes with parameters $[32,16,8].$ We summarise the results in a table. 

\begin{table}[h!]
\centering
\begin{tabular}{cccc}
\hline
$\mathcal{C}_i$ & First row of $A$    & First row of $B$    & $|Aut(\mathcal{C}_i)|$            \\ \hline
$\mathcal{C}_1$ & $(0,0,0,0,0,1,0,1)$ & $(0,0,0,1,1,1,1,1)$ & $2^{15}\cdot 3^2 \cdot 5 \cdot 7$ \\ \hline
$\mathcal{C}_2$ & $(0,0,0,0,0,1,1,1)$ & $(0,1,0,1,1,1,1,1)$ & $2^{15}\cdot 3^2$                 \\ \hline
$\mathcal{C}_3$ & $(0,0,0,0,1,1,1,1)$ & $(0,0,0,1,0,0,1,1)$ & $2^5 \cdot 3 \cdot 5 \cdot 31$    \\ \hline
\end{tabular}
\end{table}
\end{exam}

\begin{exam}
We now amend $\sigma(v)$ from the previous example by forming a composite matrix. Let $G=\langle x ,y \ | \ x^8=y^2=1,x^y=x^{-1} \rangle \cong D_{16}$ and $v=\sum_{i=0}^7 \sum_{j=0}^1 \alpha_{1+i+8j}x^iy^j \in \mathbb{F}_2D_{16}.$ Also let $H_1=\langle a,b \ | \ a^2=b^2=1, ab=ba \rangle \cong C_4 \times C_2$ and $H_2=\langle c,d \ | \ c^4=d^2=c^d=c^{-1} \rangle \cong D_8.$ Now we define the composite matrix as:

$$\Omega(v)=\begin{pmatrix}
A_1'&A_2'\\
A_3'&A_4'
\end{pmatrix},$$
where

\resizebox{0.65\textwidth}{!}{\begin{minipage}{\textwidth}
$$A_1'=\begin{pmatrix}
g_1^{-1}g_1&g_1^{-1}g_{2}&g_1^{-1}g_{3}&g_1^{-1}g_{4}&g_1^{-1}g_{5}&g_1^{-1}g_{6}&g_1^{-1}g_{7}&g_1^{-1}g_{8}\\
\phi_1((h_1)_2^{-1}(h_1)_1)&\phi_1((h_1)_2^{-1}(h_1)_2)&\phi_1((h_1)_2^{-1}(h_1)_3)&\phi_1((h_1)_2^{-1}(h_1)_4)&\phi_1((h_1)_2^{-1}(h_1)_5)&\phi_1((h_1)_2^{-1}(h_1)_6)&\phi_1((h_1)_2^{-1}(h_1)_7)&\phi_1((h_1)_2^{-1}(h_1)_8)\\
\phi_1((h_1)_3^{-1}(h_1)_1)&\phi_1((h_1)_3^{-1}(h_1)_2)&\phi_1((h_1)_3^{-1}(h_1)_3)&\phi_1((h_1)_3^{-1}(h_1)_4)&\phi_1((h_1)_3^{-1}(h_1)_5)&\phi_1((h_1)_3^{-1}(h_1)_6)&\phi_1((h_1)_3^{-1}(h_1)_7)&\phi_1((h_1)_3^{-1}(h_1)_8)\\
\phi_1((h_1)_4^{-1}(h_1)_1)&\phi_1((h_1)_4^{-1}(h_1)_2)&\phi_1((h_1)_4^{-1}(h_1)_3)&\phi_1((h_1)_4^{-1}(h_1)_4)&\phi_1((h_1)_4^{-1}(h_1)_5)&\phi_1((h_1)_4^{-1}(h_1)_6)&\phi_1((h_1)_4^{-1}(h_1)_7)&\phi_1((h_1)_4^{-1}(h_1)_8)\\
\phi_1((h_1)_5^{-1}(h_1)_1)&\phi_1((h_1)_5^{-1}(h_1)_2)&\phi_1((h_1)_5^{-1}(h_1)_3)&\phi_1((h_1)_5^{-1}(h_1)_4)&\phi_1((h_1)_5^{-1}(h_1)_5)&\phi_1((h_1)_5^{-1}(h_1)_6)&\phi_1((h_1)_5^{-1}(h_1)_7)&\phi_1((h_1)_5^{-1}(h_1)_8)\\
\phi_1((h_1)_6^{-1}(h_1)_1)&\phi_1((h_1)_6^{-1}(h_1)_2)&\phi_1((h_1)_6^{-1}(h_1)_3)&\phi_1((h_1)_6^{-1}(h_1)_4)&\phi_1((h_1)_6^{-1}(h_1)_5)&\phi_1((h_1)_6^{-1}(h_1)_6)&\phi_1((h_1)_6^{-1}(h_1)_7)&\phi_1((h_1)_6^{-1}(h_1)_8)\\
\phi_1((h_1)_7^{-1}(h_1)_1)&\phi_1((h_1)_7^{-1}(h_1)_2)&\phi_1((h_1)_7^{-1}(h_1)_3)&\phi_1((h_1)_7^{-1}(h_1)_4)&\phi_1((h_1)_7^{-1}(h_1)_5)&\phi_1((h_1)_7^{-1}(h_1)_6)&\phi_1((h_1)_7^{-1}(h_1)_7)&\phi_1((h_1)_7^{-1}(h_1)_8)\\
\phi_1((h_1)_8^{-1}(h_1)_1)&\phi_1((h_1)_8^{-1}(h_1)_2)&\phi_1((h_1)_8^{-1}(h_1)_3)&\phi_1((h_1)_8^{-1}(h_1)_4)&\phi_1((h_1)_8^{-1}(h_1)_5)&\phi_1((h_1)_8^{-1}(h_1)_6)&\phi_1((h_1)_8^{-1}(h_1)_7)&\phi_1((h_1)_8^{-1}(h_1)_8)\\
\end{pmatrix},$$
\end{minipage}}

\resizebox{0.65\textwidth}{!}{\begin{minipage}{\textwidth}
$$A_2'=\begin{pmatrix}
g_1^{-1}g_9&g_1^{-1}g_{10}&g_1^{-1}g_{11}&g_1^{-1}g_{12}&g_1^{-1}g_{13}&g_1^{-1}g_{14}&g_1^{-1}g_{15}&g_1^{-1}g_{16}\\
\phi_2((h_1)_2^{-1}(h_1)_1)&\phi_2((h_1)_2^{-1}(h_1)_2)&\phi_2((h_1)_2^{-1}(h_1)_3)&\phi_2((h_1)_2^{-1}(h_1)_4)&\phi_2((h_1)_2^{-1}(h_1)_5)&\phi_2((h_1)_2^{-1}(h_1)_6)&\phi_2((h_1)_2^{-1}(h_1)_7)&\phi_2((h_1)_2^{-1}(h_1)_8)\\
\phi_2((h_1)_3^{-1}(h_1)_1)&\phi_2((h_1)_3^{-1}(h_1)_2)&\phi_2((h_1)_3^{-1}(h_1)_3)&\phi_2((h_1)_3^{-1}(h_1)_4)&\phi_2((h_1)_3^{-1}(h_1)_5)&\phi_2((h_1)_3^{-1}(h_1)_6)&\phi_2((h_1)_3^{-1}(h_1)_7)&\phi_2((h_1)_3^{-1}(h_1)_8)\\
\phi_2((h_1)_4^{-1}(h_1)_1)&\phi_2((h_1)_4^{-1}(h_1)_2)&\phi_2((h_1)_4^{-1}(h_1)_3)&\phi_2((h_1)_4^{-1}(h_1)_4)&\phi_2((h_1)_4^{-1}(h_1)_5)&\phi_2((h_1)_4^{-1}(h_1)_6)&\phi_2((h_1)_4^{-1}(h_1)_7)&\phi_2((h_1)_4^{-1}(h_1)_8)\\
\phi_2((h_1)_5^{-1}(h_1)_1)&\phi_2((h_1)_5^{-1}(h_1)_2)&\phi_2((h_1)_5^{-1}(h_1)_3)&\phi_2((h_1)_5^{-1}(h_1)_4)&\phi_2((h_1)_5^{-1}(h_1)_5)&\phi_2((h_1)_5^{-1}(h_1)_6)&\phi_2((h_1)_5^{-1}(h_1)_7)&\phi_2((h_1)_5^{-1}(h_1)_8)\\
\phi_2((h_1)_6^{-1}(h_1)_1)&\phi_2((h_1)_6^{-1}(h_1)_2)&\phi_2((h_1)_6^{-1}(h_1)_3)&\phi_2((h_1)_6^{-1}(h_1)_4)&\phi_2((h_1)_6^{-1}(h_1)_5)&\phi_2((h_1)_6^{-1}(h_1)_6)&\phi_2((h_1)_6^{-1}(h_1)_7)&\phi_2((h_1)_6^{-1}(h_1)_8)\\
\phi_2((h_1)_7^{-1}(h_1)_1)&\phi_2((h_1)_7^{-1}(h_1)_2)&\phi_2((h_1)_7^{-1}(h_1)_3)&\phi_2((h_1)_7^{-1}(h_1)_4)&\phi_2((h_1)_7^{-1}(h_1)_5)&\phi_2((h_1)_7^{-1}(h_1)_6)&\phi_2((h_1)_7^{-1}(h_1)_7)&\phi_2((h_1)_7^{-1}(h_1)_8)\\
\phi_2((h_1)_8^{-1}(h_1)_1)&\phi_2((h_1)_8^{-1}(h_1)_2)&\phi_2((h_1)_8^{-1}(h_1)_3)&\phi_2((h_1)_8^{-1}(h_1)_4)&\phi_2((h_1)_8^{-1}(h_1)_5)&\phi_2((h_1)_8^{-1}(h_1)_6)&\phi_2((h_1)_8^{-1}(h_1)_7)&\phi_2((h_1)_8^{-1}(h_1)_8)\\
\end{pmatrix},$$
\end{minipage}}

\resizebox{0.65\textwidth}{!}{\begin{minipage}{\textwidth}
$$A_3'=\begin{pmatrix}
g_9^{-1}g_1&g_9^{-1}g_{2}&g_9^{-1}g_{3}&g_9^{-1}g_{4}&g_9^{-1}g_{5}&g_9^{-1}g_{6}&g_9^{-1}g_{7}&g_9^{-1}g_{8}\\
\phi_3((h_2)_2^{-1}(h_2)_1)&\phi_3((h_2)_2^{-1}(h_2)_2)&\phi_3((h_2)_2^{-1}(h_2)_3)&\phi_3((h_2)_2^{-1}(h_2)_4)&\phi_3((h_2)_2^{-1}(h_2)_5)&\phi_3((h_2)_2^{-1}(h_2)_6)&\phi_3((h_1)_2^{-1}(h_2)_7)&\phi_3((h_2)_2^{-1}(h_2)_8)\\
\phi_3((h_2)_3^{-1}(h_2)_1)&\phi_3((h_2)_3^{-1}(h_2)_2)&\phi_3((h_2)_3^{-1}(h_2)_3)&\phi_3((h_2)_3^{-1}(h_2)_4)&\phi_3((h_2)_3^{-1}(h_2)_5)&\phi_3((h_2)_3^{-1}(h_2)_6)&\phi_3((h_2)_3^{-1}(h_2)_7)&\phi_3((h_2)_3^{-1}(h_2)_8)\\
\phi_3((h_2)_4^{-1}(h_2)_1)&\phi_3((h_2)_4^{-1}(h_2)_2)&\phi_3((h_2)_4^{-1}(h_2)_3)&\phi_3((h_2)_4^{-1}(h_2)_4)&\phi_3((h_2)_4^{-1}(h_2)_5)&\phi_3((h_2)_4^{-1}(h_2)_6)&\phi_3((h_2)_4^{-1}(h_2)_7)&\phi_3((h_2)_4^{-1}(h_2)_8)\\
\phi_3((h_2)_5^{-1}(h_2)_1)&\phi_3((h_2)_5^{-1}(h_2)_2)&\phi_3((h_2)_5^{-1}(h_2)_3)&\phi_3((h_2)_5^{-1}(h_2)_4)&\phi_3((h_2)_5^{-1}(h_2)_5)&\phi_3((h_2)_5^{-1}(h_2)_6)&\phi_3((h_2)_5^{-1}(h_2)_7)&\phi_3((h_2)_5^{-1}(h_2)_8)\\
\phi_3((h_2)_6^{-1}(h_2)_1)&\phi_3((h_2)_6^{-1}(h_2)_2)&\phi_3((h_2)_6^{-1}(h_2)_3)&\phi_3((h_2)_6^{-1}(h_2)_4)&\phi_3((h_2)_6^{-1}(h_2)_5)&\phi_3((h_2)_6^{-1}(h_2)_6)&\phi_3((h_2)_6^{-1}(h_2)_7)&\phi_3((h_2)_6^{-1}(h_2)_8)\\
\phi_3((h_2)_7^{-1}(h_2)_1)&\phi_3((h_2)_7^{-1}(h_2)_2)&\phi_3((h_2)_7^{-1}(h_2)_3)&\phi_3((h_2)_7^{-1}(h_2)_4)&\phi_3((h_2)_7^{-1}(h_2)_5)&\phi_3((h_2)_7^{-1}(h_2)_6)&\phi_3((h_2)_7^{-1}(h_2)_7)&\phi_3((h_2)_7^{-1}(h_2)_8)\\
\phi_3((h_2)_8^{-1}(h_2)_1)&\phi_3((h_2)_8^{-1}(h_2)_2)&\phi_3((h_2)_8^{-1}(h_2)_3)&\phi_3((h_2)_8^{-1}(h_2)_4)&\phi_3((h_2)_8^{-1}(h_2)_5)&\phi_3((h_2)_8^{-1}(h_2)_6)&\phi_3((h_2)_8^{-1}(h_2)_7)&\phi_3((h_2)_8^{-1}(h_2)_8)\\
\end{pmatrix},$$
\end{minipage}}

\resizebox{0.65\textwidth}{!}{\begin{minipage}{\textwidth}
$$A_4'=\begin{pmatrix}
g_9^{-1}g_9&g_9^{-1}g_{10}&g_9^{-1}g_{11}&g_9^{-1}g_{12}&g_9^{-1}g_{13}&g_9^{-1}g_{14}&g_9^{-1}g_{15}&g_9^{-1}g_{16}\\
\phi_4((h_2)_2^{-1}(h_2)_1)&\phi_4((h_2)_2^{-1}(h_2)_2)&\phi_4((h_2)_2^{-1}(h_2)_3)&\phi_4((h_2)_2^{-1}(h_2)_4)&\phi_4((h_2)_2^{-1}(h_2)_5)&\phi_4((h_2)_2^{-1}(h_2)_6)&\phi_4((h_1)_2^{-1}(h_2)_7)&\phi_4((h_2)_2^{-1}(h_2)_8)\\
\phi_4((h_2)_3^{-1}(h_2)_1)&\phi_4((h_2)_3^{-1}(h_2)_2)&\phi_4((h_2)_3^{-1}(h_2)_3)&\phi_4((h_2)_3^{-1}(h_2)_4)&\phi_4((h_2)_3^{-1}(h_2)_5)&\phi_4((h_2)_3^{-1}(h_2)_6)&\phi_4((h_2)_3^{-1}(h_2)_7)&\phi_4((h_2)_3^{-1}(h_2)_8)\\
\phi_4((h_2)_4^{-1}(h_2)_1)&\phi_4((h_2)_4^{-1}(h_2)_2)&\phi_4((h_2)_4^{-1}(h_2)_3)&\phi_4((h_2)_4^{-1}(h_2)_4)&\phi_4((h_2)_4^{-1}(h_2)_5)&\phi_4((h_2)_4^{-1}(h_2)_6)&\phi_4((h_2)_4^{-1}(h_2)_7)&\phi_4((h_2)_4^{-1}(h_2)_8)\\
\phi_4((h_2)_5^{-1}(h_2)_1)&\phi_4((h_2)_5^{-1}(h_2)_2)&\phi_4((h_2)_5^{-1}(h_2)_3)&\phi_4((h_2)_5^{-1}(h_2)_4)&\phi_4((h_2)_5^{-1}(h_2)_5)&\phi_4((h_2)_5^{-1}(h_2)_6)&\phi_4((h_2)_5^{-1}(h_2)_7)&\phi_4((h_2)_5^{-1}(h_2)_8)\\
\phi_4((h_2)_6^{-1}(h_2)_1)&\phi_4((h_2)_6^{-1}(h_2)_2)&\phi_4((h_2)_6^{-1}(h_2)_3)&\phi_4((h_2)_6^{-1}(h_2)_4)&\phi_4((h_2)_6^{-1}(h_2)_5)&\phi_4((h_2)_6^{-1}(h_2)_6)&\phi_4((h_2)_6^{-1}(h_2)_7)&\phi_4((h_2)_6^{-1}(h_2)_8)\\
\phi_4((h_2)_7^{-1}(h_2)_1)&\phi_4((h_2)_7^{-1}(h_2)_2)&\phi_4((h_2)_7^{-1}(h_2)_3)&\phi_4((h_2)_7^{-1}(h_2)_4)&\phi_4((h_2)_7^{-1}(h_2)_5)&\phi_4((h_2)_7^{-1}(h_2)_6)&\phi_4((h_2)_7^{-1}(h_2)_7)&\phi_4((h_2)_7^{-1}(h_2)_8)\\
\phi_4((h_2)_8^{-1}(h_2)_1)&\phi_4((h_2)_8^{-1}(h_2)_2)&\phi_4((h_2)_8^{-1}(h_2)_3)&\phi_4((h_2)_8^{-1}(h_2)_4)&\phi_4((h_2)_8^{-1}(h_2)_5)&\phi_4((h_2)_8^{-1}(h_2)_6)&\phi_4((h_2)_8^{-1}(h_2)_7)&\phi_4((h_2)_8^{-1}(h_2)_8)\\
\end{pmatrix},$$
\end{minipage}}\\

and where:

\begin{table}[h!]
\centering
\begin{tabular}{cclcc}
\multirow{2}{*}{$\phi_1:$} & $(h_1)_i \xrightarrow{\phi_1} \alpha_{g_1^{-1}g_i}$ &  & \multirow{2}{*}{$\phi_2:$} & $(h_1)_i \xrightarrow{\phi_2} \alpha_{g_1^{-1}g_i}$    \\
                           & $\text{for} \ i=\{1,2,\dots,8\}$               &  &                            & $\text{for when} \ i=\{9,10,\dots,16\}$
\end{tabular}
\end{table}

\begin{table}[h!]
\centering
\begin{tabular}{cclcc}
\multirow{2}{*}{$\phi_3:$} & $(h_2)_i \xrightarrow{\phi_3} \alpha_{g_9^{-1}g_i}$ &  & \multirow{2}{*}{$\phi_4:$} & $(h_2)_i \xrightarrow{\phi_4} \alpha_{g_9^{-1}g_i}$    \\
                           & $\text{for} \ i=\{1,2,\dots,8\}$               &  &                            & $\text{for when} \ i=\{9,10,\dots,16\}.$
\end{tabular}
\end{table}

This results in a composite matrix of the following form:

$$\Omega(v)=\begin{pmatrix}
\begin{tabular}{cc|cc}
$A_1$ & $B_1$ & $A_2$ & $B_2$ \\
$B_1$ & $A_1$ & $B_2$ & $A_2$ \\ \hline
$A_3$ & $B_3$ & $A_4$ & $B_4$ \\
$B_3^T$ & $A_3^T$ & $B_4^T$ & $A_4^T$
\end{tabular}
\end{pmatrix},$$
where 
\begin{eqnarray*} 
A_1&=&circ(\alpha_1,\alpha_2,\alpha_3,\alpha_4),\\ B_1&=&circ(\alpha_5,\alpha_6,\alpha_7,\alpha_8), \\ A_2&=&circ(\alpha_9,\alpha_{10},\alpha_{11},\alpha_{12}), \\B_2&=&circ(\alpha_{13},\alpha_{14},\alpha_{15},\alpha_{16}), \\ A_3&=&circ(\alpha_9,\alpha_{16},\alpha_{15},\alpha_{14}),\\ B_3&=&circ(\alpha_{13},\alpha_{12},\alpha_{11},\alpha_{10}), \\ A_4&=&circ(\alpha_1,\alpha_8,\alpha_7,\alpha_6), \\ 
 B_4&=&circ(\alpha_5,\alpha_4,\alpha_3,\alpha_2) \end{eqnarray*}  and where $\alpha_i \in \mathbb{F}_2.$  We now employ the generator matrix of the form $[I_{16} \ | \ \Omega(v)]$ to search for binary self-dual codes with parameters $[32,16,8].$ We summarise the results in a table.

\begin{table}[h!]
\begin{center}
\scalebox{0.8}{
\begin{tabular}{cccccccccc}
\hline
$\mathcal{C}_i$ & $r_{A_1}$   & $r_{B_1}$   & $r_{A_2}$   & $r_{B_2}$   & $r_{A_3}$   & $r_{B_3}$   & $r_{A_4}$   & $r_{B_4}$   & $|Aut(\mathcal{C}_i)|$  \\ \hline
$\mathcal{C}_1$ & $(0,0,1,0)$ & $(0,0,1,0)$ & $(0,0,1,0)$ & $(1,1,1,1)$ & $(0,1,1,1)$ & $(1,0,1,0)$ & $(0,0,1,0)$ & $(0,0,1,0)$ & $2^9 \cdot 3^2 \cdot 5$ \\ \hline
\end{tabular}}
\end{center}
\end{table}
\end{exam}

The order of the automorphism group of the code obtained in Example 7 is different from the order of automorphism of codes obtained in Example 6. This shows that the composite matrices can be used to produce codes whose structure is not attainable from matrices of the form $[I_n \ | \ \sigma(v)]$ or other classical techniques for producing extremal binary self-dual codes. In fact, this is the main motivating factor for this construction, that is, we construct codes whose automorphism group differs from other constructions which means we find codes that are inaccessible from other techniques.

\begin{thm}
Let $R$ be a finite commutative Frobenius ring, $G$ be a group of order $n$ and $H_i$ be finite groups of order $r$ such that $r$ is a factor of $n$ with $n>1$ and $n,r \neq1.$ Let $v \in RG$ and let $\Omega(v)$ be the corresponding composite matrix over $R.$ The matrix $G=[I_n \ | \ \Omega(v)]$ generates a self-dual code $\mathcal{C}$ over $R$ if and only if $\Omega(v)\Omega(v)^T=-I_n.$

\begin{proof}
The code $\mathcal{C}$ is self-dual if and only if $GG^T$ is the zero matrix over $R.$ Now,

$$GG^T=[I_n \ | \ \Omega(v)][I_n \ | \ \Omega(v)]^T=[I_n \ | \ \Omega(v)\Omega(v)^T].$$
Thus, $GG^T$ is the zero matrix over $R$ if and only if $\Omega(v)\Omega(v)^T=-I_n.$
\end{proof}
\end{thm}

We saw earlier in the work that $\Omega(v^T)=\Omega(v)^T.$ Now using Theorem 7.1, the fact that $\Omega$ is a ring homomorphism, and the fact that $\Omega(v)=-I_n$ if and only if $v=-1,$ we get the following corollary.

\begin{cor}
Let $R$ be a finite commutative Frobenius ring, $G$ be a group of order $n$ and $H_i$ be finite groups of order $r$ such that $r$ is a factor of $n$ with $n>1$ and $n,r \neq1.$ Let $v \in RG$ and let $\Omega(v)$ be the corresponding composite matrix over $R.$ The matrix $[I_n \ | \ \Omega(v)]$ generates a self-dual code over $R$ if and only if $vv^T=-1.$ In particular $v$ has to be a unit.
\end{cor}

When we consider a ring of characteristic 2, we have $-I_n=I_n,$ which leads to the following further important result:

\begin{cor}
Let $R$ be a finite commutative Frobenius ring of characteristic 2, $G$ be a group of order $n$ and $H_i$ be finite groups of order $r$ such that $r$ is a factor of $n$ with $n>1$ and $n,r \neq1.$ Let $v \in RG$ and let $\Omega(v)$ be the corresponding composite matrix over $R.$ Then the matrix $[I_n \ | \ \Omega(v)]$ generates a self-dual code over $R$ if and only if $v$ satisfies $vv^t=1,$ namely $v$ is a unitary unit in $RG.$
\end{cor}

\subsection{New Extremal Self-Dual Binary Codes of Length $68$}

In this section, we apply the generator matrix of the form $[I \ | \ \Omega(v)]$ over the ring $\mathbb{F}_4+u\mathbb{F}_4$ to find extremal self-dual codes whose binary images are the extremal self-dual binary codes of length 64. We then apply a very well known extension method to obtain codes of length 68. We next apply a very recent technique, called a neighbor of a neighbor method, to find a family of neighbours which turn out to be extremal self-dual binary codes of length 68 with parameters not known in the literature before. In particular we find new codes of length 68 with the rare parameters of $\gamma=7, 8, 9.$ We split this section into the following subsections. In the first one, we describe the ring $\mathbb{F}_4+u\mathbb{F}_4$ and give the most up to date list of codes of length 68 with parameters known in the literature. Then we define the generator matrix of the form $[I \ | \ \Omega(v)]$,  which we use to find codes of length 64. We then extend these codes to obtain codes of length 68. Finally, we apply the family of neighbours method to find codes of length 68 with parameters not known in the literature.

\subsubsection{The Ring $\mathbb{F}_4+u\mathbb{F}_4$, the Extension and Neighbour Methods}

Let us recall the following Gray Maps from
\cite{VI} and \cite{V}:

\begin{table}[h!]
\centering
\label{my-label}
\begin{tabular}{l|l}
$\psi_{\mathbb{F}_4} : (\mathbb{F}_4)^n \to (\mathbb{F}_2)^{2n}$ & $\varphi_{%
\mathbb{F}_2+u\mathbb{F}_2} : (\mathbb{F}_2+u\mathbb{F}_2)^n \to \mathbb{F}%
_2^{2n}$ \\
$a\omega +b\overline{\omega} \mapsto (a,b) , \ a,b \in \mathbb{F}_2^n$ & $%
a+bu \mapsto (b,a+b), \ a,b \in \mathbb{F}_2^{n}.$%
\end{tabular}%
\end{table}

In \cite{VII}, these maps were generalized to the following Gray maps:
\begin{table}[h!]
\centering
\label{my-label}
\begin{tabular}{l|l}
$\psi_{\mathbb{F}_4+u\mathbb{F}_4} : (\mathbb{F}_4+u\mathbb{F}_4)^n \to (%
\mathbb{F}_2+u\mathbb{F}_2)^{2n}$ & $\varphi_{\mathbb{F}_4+u\mathbb{F}_4} : (%
\mathbb{F}_4+u\mathbb{F}_4)^n \to \mathbb{F}_4^{2n}$ \\
$a\omega +b\overline{\omega} \mapsto (a,b) , \ a,b \in (\mathbb{F}_2+u%
\mathbb{F}_2)^n$ & $a+bu \mapsto (b,a+b), \ a,b \in \mathbb{F}_4^{n}.$%
\end{tabular}%
\end{table}

\begin{prop}
(\cite{VII}) Let $C$ be a code over $\mathbb{F}_4+u\mathbb{F}_4.$ If $C$ is
self-orthogonal, then so are $\psi_{\mathbb{F}_4+u\mathbb{F}_4}(C)$ and $%
\varphi_{\mathbb{F}_4+u\mathbb{F}_4}(C)$. The code $C$ is a Type~I (resp.
Type~II) code over $\mathbb{F}_4+u\mathbb{F}_4$ if and only if $\varphi_{%
\mathbb{F}_4+u\mathbb{F}_4}(C)$ is a Type~I (resp. Type~II) $\mathbb{F}_4$%
-code, if and only if $\psi_{\mathbb{F}_4+u\mathbb{F}_4}(C)$ is a Type~I
(resp. Type~II) $\mathbb{F}_2+u\mathbb{F}_2$-code. Furthermore, the minimum
Lee weight of $C$ is the same as the minimum Lee weight of $\psi_{\mathbb{F}%
_4+u\mathbb{F}_4}(C)$ and $\varphi_{\mathbb{F}_4+u\mathbb{F}_4}(C)$.
\end{prop}

The next corollary follows immediately from the proposition and we will use
this result repeatedly to produce binary codes.

\begin{cor}
Suppose that $C$ is a self-dual code over $\mathbb{F}_4+u\mathbb{F}_4$ of
length $n$ and minimum Lee distance $d$. Then $\varphi_{\mathbb{F}_2+u%
\mathbb{F}_2} \circ \psi_{\mathbb{F}_4+u\mathbb{F}_4}(C)$ is a binary $%
[4n,2n,d]$ self-dual code. Moreover, the Lee weight enumerator of $C$ is
equal to the Hamming weight enumerator of $\varphi_{\mathbb{F}_2+u\mathbb{F}%
_2} \circ \psi_{\mathbb{F}_4+u\mathbb{F}_4}(C).$ If $C$ is Type~I (Type~II),
then so is $\varphi_{\mathbb{F}_2+u\mathbb{F}_2} \circ \psi_{\mathbb{F}_4+u%
\mathbb{F}_4}(C).$
\end{cor}

For the computational results in later sections, we are going to use the following extension method to obtain codes of length $n+2.$ 

\begin{thm}\label{extensionthm}
(\cite{XIV}) Let $\mathcal{C}$ be a self-dual code of length $n$
over a commutative\ Frobenius ring with identity $R$ and $G=(r_i)$ be a $k \times n$ generator matrix for $\mathcal{C}$, where $r_i
$ is the i-th row of $G $, $1\leq i \leq k.$ Let $c$ be a unit in $R$ such
that $c^2=-1$ and $X$ be a vector in $S^n$ with $\langle X,X \rangle=-1.$
Let $y_i=\langle r_i, X \rangle $ for $1 \leq i \leq k.$ The following matrix

\begin{equation*}
\begin{bmatrix}
\begin{tabular}{cc|c}
$1$ & $0$ & $X$ \\ \hline
$y_1$ & $cy_1$ & $r_1$ \\
$\vdots$ & $\vdots$ & $\vdots$ \\
$y_k$ & $cy_k$ & $r_k$%
\end{tabular}%
\end{bmatrix}%
,
\end{equation*}
generates a self-dual code $\mathcal{D}$ over $R$ of length $n+2.$
\end{thm}

We will also apply the neighbor method and its generalization to search for new extremal binary self-dual codes from codes obtained directly from our constructions or from the described above, extension method. Two self-dual binary codes of length $2n$ are said to be neighbors
of each
other if their intersection has dimension $n-1$. Let $x\in {\mathbb{F}}%
_{2}^{2n}-\mathcal{C}$ then $\mathcal{D}=\left\langle \left\langle
x\right\rangle ^{\bot }\cap \mathcal{C},x\right\rangle $ is a neighbour of $%
\mathcal{C}$.

Recently in \cite{GN}, the neighbor method has been extended and the following formula for constructing the $k^{th}$-range neighbour codes was provided:

$$\mathcal{N}_{(i+1)}=\left\langle \left\langle x_i \right\rangle^{\bot} \cap \mathcal{N}_{(i)}, x_i \right\rangle,$$

where $\mathcal{N}_{(i+1)}$ is the neighbour of $\mathcal{N}_{(i)}$ and $x_i \in \mathbb{F}_2^{2n}-\mathcal{N}_{(i)}.$

There are two possibilities for the weight enumerators of extremal singly-even $[64,32,12]_2$ codes (\cite{IX}):

\begin{equation*}
W_{64,1}=1+(1312+16\beta)y^{12}+(22016-64\beta)y^{14}+\dots, \ 14 \leq \beta
\leq 284,
\end{equation*}
\begin{equation*}
W_{64,2}=1+(1312+16\beta)y^{12}+(23040-64\beta)y^{14}+ \dots, \ 0 \leq \beta
\leq 277.
\end{equation*}
Recently, many new codes are constructed for both weight enumerators in \cite{Hankel}, \cite{Kaya} and \cite{Yankov3}. With the most updated information, the existence
of codes is known for $\beta =$14, 16, 18, 19, 20, 22, 24, 25, 26, 28, 29,
30, 32, 34, 35, 36, 38, 39, 44, 46, 49, 53, 54, 58, 59, 60, 64 and $74$ in $%
W_{64,1}$ and for $\beta =0,\dots ,$40, 41, 42, 44, 45, 46, 47, 48, 49, 50,
51, 52, 54, 55, 56, 57, 58, 60, 62, 64, 69, 72, 80, 88, 96, 104, 108, 112,
114, 118, 120 and $184$ in $W_{64,2}.$

The weight
enumerator of a self-dual $[68,34,12]_2$ code is in one of the following
forms by \cite{VIII, XIII}:
\begin{equation*}
W_{68,1}=1+(442+4\beta)y^{12}+(10864-8\beta)y^{14}+\dots,
\end{equation*}
\begin{equation*}
W_{68,2}=1+(442+4\beta)y^{12}+(14960-8\beta-256\gamma)y^{14}+\dots \ ,
\end{equation*}
where $\beta$ and $\gamma$ are parameters and $0 \leq \gamma \leq 9.$ The
first examples of codes with a $\gamma =7$ in $W_{68,2}$ are constructed in
\cite{Yankov2}. The first examples of codes with $\gamma=8,9$ in $W_{68,2}$ are constructed in \cite{GN}. Together with these the existence of the codes in $W_{68,2}$
is known for the following parameters (see \cite{XII, X, XI, Yankov1,
Yankov2, Hankel}):\newline

$%
\begin{array}{l}
\gamma =0,\ \beta \in \{2m|m=0,7,11,14,17,20,21,\dots
,99,100,102,105,110,119,136,165\};\ \text{or} \\
\qquad \beta \in \{2m+1|m=3,5,8,10,15,16,17,19,20,\dots ,82,87,91,\dots,99,101,104,110,\\
\qquad 115\};
\\
\gamma =1,\ \beta \in \{2m|m=19,22,\dots ,99,108\};\ \text{or}\\
\qquad \beta \in
\{2m+1|m=24,\dots ,85,94,100,101,106,108,116\}; \\
\gamma =2,\ \beta \in \{2m|m=29,\dots ,100,103,104\};\ \text{or}\

\beta \in\{2m+1|m=32,\dots ,81,84,85,86\}; \\
\gamma =3,\ \beta \in \{2m|m=39,\dots ,92,94,95,97,98,101,102\};\ \text{or}
\\
\qquad \beta \in \{2m+1|m=38,39,40,42,43,\dots ,77,79,80,81,83,87,88,89,96\}; \\
\gamma =4,\ \beta \in \{2m|m=43,46,\dots ,58,60,\dots ,93,97,98,100\};\
\text{or} \\
\qquad \beta \in \{2m+1|m=48,\dots ,55,57,58,60,61,62,64,68,\ldots
,72,74,78,79,80,83,84,\\
\qquad 85,89,95\}; \\
\gamma =5\ \beta \in \{m|m=113,116,\dots ,153,158,\dots
,169,182,187,189,191,193,195,198,\\
\qquad 200,202,211\}; \\
\gamma =6\ \text{with}\ \beta \in \{2m|m=69,77,78,79,81,88,91,93,94,95,97,\dots,103\}; \\
\qquad \text{or}\ \beta \in \{2m+1|m=87,\dots,100,103\};\\
\gamma =7\ \text{with}\ \beta \in \{2m|m=49,56,63,70,77,83,\dots,99,105,106,112,119,126,133,147\};\\
\qquad \text{or}\ \beta \in \{2m+1|m=52,59,66,73,80,81,84,\dots,99,101,108,115,122,129,136\};\\
\gamma =8\ \text{with}\ \beta \in \{2m|m=90,\dots,110\};\ \text{or}\ \beta \in \{2m+1|m=90,\dots,110\};\\
\gamma =9\ \text{with}\ \beta \in \{2m|m=93,\dots,105,107,\dots,115\};\ \text{or}\ \beta \in \{2m+1|m=93,94,96,97,99,\\
\qquad \dots,112\};\\
\end{array}%
$\newline

\subsubsection{The Generator Matrix}

We now define the generator matrix of the form $[I \ | \ \Omega(v)]$ which we then employ to search for self-dual codes over the ring $\mathbb{F}_4+u\mathbb{F}_4.$ Of course, $I$ is simply the identity matrix so we define $\Omega(v).$

Let $G=\langle x,y \ | \ x^4=y^2=1, x^y=x^{-1} \rangle \cong D_8.$ Let $v=\alpha_1+\alpha_{x}x+\alpha_{x^2}x^2+\alpha_{x^3}x^3+\alpha_{y}y+\alpha_{xy}xy+\alpha_{x^2y}x^2y+\alpha_{x^3y}x^3y \in RD_8,$ where $\alpha_{g_i} \in R.$ Let $H_1=\langle a,b \ | \ a^2=b^2=1, ab=ba \rangle \cong C_2 \times C_2$ and $H_2=\langle c \ | \ c^4=1 \rangle \cong C_4.$ We now define $\Omega(v)$ as:
$$\Omega(v)=\begin{pmatrix}
A_1'&A_2'\\
A_3'&A_4'
\end{pmatrix}=$$

\resizebox{0.7\textwidth}{!}{\begin{minipage}{\textwidth}
$$\begin{pmatrix}
\begin{tabular}{cccc|cccc}
$\alpha_{{g_1^{-1}g_1}}$ & $\alpha_{{g_1^{-1}g_2}}$ & $\alpha_{{g_1^{-1}g_3}}$ & $\alpha_{{g_1^{-1}g_4}}$ & $\alpha_{{g_1^{-1}g_5}}$ & $\alpha_{{g_1^{-1}g_6}}$ & $\alpha_{{g_1^{-1}g_7}}$ & $\alpha_{{g_1^{-1}g_8}}$ \\
$\alpha_{\phi_1((h_2)_2^{-1}(h_2)_1)}$                        & $\alpha_{\phi_1((h_2)_2^{-1}(h_2)_2)}$                        & $\alpha_{\phi_1((h_2)_2^{-1}(h_2)_3)}$                        & $\alpha_{\phi_1((h_2)_2^{-1}(h_2)_4)}$                        & $\alpha_{\phi_2((h_1)_2^{-1}(h_1)_1)}$                        & $\alpha_{\phi_2((h_1)_2^{-1}(h_1)_2)}$                        & $\alpha_{\phi_2((h_1)_2^{-1}(h_1)_3)}$                        & $\alpha_{\phi_2((h_1)_2^{-1}(h_1)_4)}$                        \\
$\alpha_{\phi_1((h_2)_3^{-1}(h_2)_1)}$                        & $\alpha_{\phi_1((h_2)_3^{-1}(h_2)_2)}$                        & $\alpha_{\phi_1(h_2)_3^{-1}(h_2)_3)}$                        & $\alpha_{\phi_1((h_2)_3^{-1}(h_2)_4)}$                        & $\alpha_{\phi_2((h_1)_3^{-1}(h_1)_1)}$                        & $\alpha_{\phi_2((h_1)_3^{-1}(h_1)_2)}$                        & $\alpha_{\phi_2((h_1)_3^{-1}(h_1)_3)}$                        & $\alpha_{\phi_2((h_1)_3^{-1}(h_1)_4)}$                        \\
$\alpha_{\phi_1((h_2)_4^{-1}(h_2)_1)}$                        & $\alpha_{\phi_1((h_2)_4^{-1}(h_2)_2)}$                        & $\alpha_{\phi_1(h_2)_4^{-1}(h_2)_3)}$                        & $\alpha_{\phi_1((h_2)_4^{-1}(h_2)_4)}$                        & $\alpha_{\phi_2((h_1)_4^{-1}(h_1)_1)}$                        & $\alpha_{\phi_2((h_1)_4^{-1}(h_1)_2)}$                        & $\alpha_{\phi_2((h_1)_4^{-1}(h_1)_3)}$                        & $\alpha_{\phi_2((h_1)_4^{-1}(h_1)_4)}$                        \\ \hline
$\alpha_{{g_5^{-1}g_1}}$ & $\alpha_{{g_5^{-1}g_2}}$ & $\alpha_{{g_5^{-1}g_3}}$ & $\alpha_{{g_5^{-1}g_4}}$ & $\alpha_{{g_5^{-1}g_5}}$ & $\alpha_{{g_5^{-1}g_6}}$ & $\alpha_{{g_5^{-1}g_7}}$ & $\alpha_{{g_5^{-1}g_8}}$ \\
$\alpha_{\phi_3((h_2)_2^{-1}(h_2)_1)}$                        & $\alpha_{\phi_3((h_2)_2^{-1}(h_2)_2)}$                        & $\alpha_{\phi_3((h_2)_2^{-1}(h_2)_3)}$                        & $\alpha_{\phi_3((h_2)_2^{-1}(h_2)_4)}$                        & $\alpha_{\phi_4((h_1)_2^{-1}(h_1)_1)}$                        & $\alpha_{\phi_4((h_1)_2^{-1}(h_1)_2)}$                        & $\alpha_{\phi_4((h_1)_2^{-1}(h_1)_3)}$                        & $\alpha_{\phi_4((h_1)_2^{-1}(h_1)_4)}$                        \\
$\alpha_{\phi_3((h_2)_3^{-1}(h_2)_1)}$                        & $\alpha_{\phi_3((h_2)_3^{-1}(h_2)_2)}$                        & $\alpha_{\phi_3(h_2)_3^{-1}(h_2)_3)}$                        & $\alpha_{\phi_3((h_2)_3^{-1}(h_2)_4)}$                        & $\alpha_{\phi_4((h_1)_3^{-1}(h_1)_1)}$                        & $\alpha_{\phi_4((h_1)_3^{-1}(h_1)_2)}$                        & $\alpha_{\phi_4((h_1)_3^{-1}(h_1)_3)}$                        & $\alpha_{\phi_4((h_1)_3^{-1}(h_1)_4)}$                        \\
$\alpha_{\phi_3((h_2)_4^{-1}(h_2)_1)}$                        & $\alpha_{\phi_3((h_2)_4^{-1}(h_2)_2)}$                        & $\alpha_{\phi_3(h_2)_4^{-1}(h_2)_3)}$                        & $\alpha_{\phi_3((h_2)_4^{-1}(h_2)_4)}$                        & $\alpha_{\phi_4((h_1)_4^{-1}(h_1)_1)}$                        & $\alpha_{\phi_4((h_1)_4^{-1}(h_1)_2)}$                        & $\alpha_{\phi_4((h_1)_4^{-1}(h_1)_3)}$                        & $\alpha_{\phi_4((h_1)_4^{-1}(h_1)_4)}$                                              
\end{tabular}
\end{pmatrix},$$
\end{minipage}}\\

where:
\begin{table}[h!]
\centering
\begin{tabular}{cclcc}
\multirow{2}{*}{$\phi_1:$} & $(h_2)_i \xrightarrow{\phi_1} g_1^{-1}g_i$ &  & \multirow{2}{*}{$\phi_2:$} & $(h_1)_i \xrightarrow{\phi_2} g_1^{-1}g_j$    \\
                           & $\text{for} \ i=\{1,2,3,4\}$               &  &                            & $\text{for when} \ \{i=1,j=5,i=2,j=6,i=3,j=7,i=4,j=8\}$ 
\end{tabular},
\end{table}

\begin{table}[h!]
\centering
\begin{tabular}{cclcc}
\multirow{2}{*}{$\phi_3:$} & $(h_1)_i \xrightarrow{\phi_1} g_5^{-1}g_i$ &  & \multirow{2}{*}{$\phi_4:$} & $(h_1)_i \xrightarrow{\phi_2} g_5^{-1}g_j$    \\
                           & $\text{for} \ i=\{1,2,3,4\}$               &  &                            & $\text{for when} \ \{i=1,j=5,i=2,j=6,i=3,j=7,i=4,j=8\}$ 
\end{tabular}.
\end{table}

in $A_1',$ $A_2',$ $A_3'$ and $A_4'.$ This results in a composite matrix over $R$ of the following form:

\begin{equation}\label{CompMatrix}
\Omega(v)=\begin{pmatrix}
\begin{tabular}{cccc|cccc}
$\alpha_1$ & $\alpha_x$ & $\alpha_{x^2}$ & $\alpha_{x^3}$ & $\alpha_{y}$ & $\alpha_{xy}$ & $\alpha_{x^2y}$ & $\alpha_{x^3y}$ \\
$\alpha_{x}$                        & $\alpha_{1}$                        & $\alpha_{x^3}$                        & $\alpha_{x^2}$                        & $\alpha_{xy}$                        & $\alpha_{y}$                        & $\alpha_{x^3y}$                        & $\alpha_{x^2y}$                        \\
$\alpha_{x^3}$                        & $\alpha_{x^2}$                        & $\alpha_{1}$                        & $\alpha_{x}$                        & $\alpha_{x^2y}$                        & $\alpha_{x^3y}$                        & $\alpha_{y}$                        & $\alpha_{xy}$                        \\
$\alpha_{x^2}$                        & $\alpha_{x^3}$                        & $\alpha_{x}$                        & $\alpha_{1}$                        & $\alpha_{x^3y}$                        & $\alpha_{x^2y}$                        & $\alpha_{xy}$                        & $\alpha_{y}$                        \\ \hline
$\alpha_{y}$ & $\alpha_{x^3y}$ & $\alpha_{x^2y}$ & $\alpha_{xy}$ & $\alpha_{1}$ & $\alpha_{x^3}$ & $\alpha_{x^2}$ & $\alpha_{x}$ \\
$\alpha_{x^3y}$                        & $\alpha_{y}$                        & $\alpha_{xy}$                        & $\alpha_{x^2y}$                        & $\alpha_{x^3}$                        & $\alpha_{1}$                        & $\alpha_{x}$                        & $\alpha_{x^2}$                                               \\
$\alpha_{x^2y}$                        & $\alpha_{xy}$                        & $\alpha_{y}$                        & $\alpha_{x^3y}$                        & $\alpha_{x^2}$                        & $\alpha_{x}$                        & $\alpha_{1}$                        & $\alpha_{x^3}$                                               \\
$\alpha_{xy}$                        & $\alpha_{x^2y}$                        & $\alpha_{x^3y}$                        & $\alpha_{y}$                        & $\alpha_{x}$                        & $\alpha_{x^2}$                        & $\alpha_{x^3}$                        & $\alpha_{1}$                                              
\end{tabular}
\end{pmatrix}.
\end{equation}

Therefore, the final form of the generator matrix which we later employ to search for self-dual codes has the following form:

\begin{equation}\label{GenMatrix}
[I \ | \ \Omega(v)],
\end{equation}

where $\Omega(v)$ is the composite matrix defined in (\ref{CompMatrix}).

\subsubsection{Computational Results}

We now employ the generator matrix defined in (\ref{GenMatrix}) over the ring $\mathbb{F}_4+u\mathbb{F}_4$ to search for codes of length 16 whose binary images are the extremal self-dual codes of length 64. In fact, we only list one of the codes found. This code in turn is used to find new extremal binary self-dual codes of length 68. All the upcoming computational results were obtained by performing the searches using MAGMA (\cite{XV}).

\begin{table}[h!]
\caption{Codes of length 64 and their $\beta$ values}\label{Codes of length 64}
\begin{center}\scalebox{1}{
\begin{tabular}{cccc}
\hline
$\mathcal{C}_i$ & $(\alpha_1,\alpha_x,\alpha_{x^2},\alpha_{x^3},\alpha_y,\alpha_{xy},\alpha_{x^2y},\alpha_{x^3y})$ & $|Aut(\mathcal{C}_i)|$ & $W_{64,2}$ \\ \hline
$1$             & $(0,w,u+1,u+1,u,wu+u,w,wu+u+1)$                                                                  & $2^4$                & $\beta=0$  \\ \hline
\end{tabular}}
\end{center}
\end{table}

We now apply Theorem 7.6 to the $\psi_{\mathbb{F}_4+u\mathbb{F}_4}$- image of the code in Table~1. As a result, we were able to find many extremal self-dual codes of length $68$ but to save space, we only list one. This code is found in Table ~2, where $1+u$ in $\mathbb{F}_2+u\mathbb{F}_2$, is denoted by $3$.

\begin{table}[h!]
\caption{Codes of length 68 from Theorem 6.6}
\begin{center}\scalebox{0.9}{
\begin{tabular}{cccccc}
\hline
$\mathcal{C}_{68,i}$ & $\mathcal{C}_i$ & $c$ & $X$                                                                 & $\gamma$ & $\beta$ \\ \hline
$\mathcal{C}_{68,1}$ & $\mathcal{C}_1$ & $1$ & $(0,3,3,u,3,1,3,3,3,3,1,1,0,3,3,1,3,1,0,u,1,3,u,3,0,1,3,u,3,0,3,1)$ & $4$      & $103$   \\ \hline
\end{tabular}}
\end{center}
\end{table}

The order of the automorphism group of the code in Table~2 is 2. We note that the code from Table 2 has parameters that are not new in the literature.

We now apply the $k^{th}$ range neighbour formula (mentioned earlier) to the code obtained in Table~2.

Let \noindent $\mathcal{N}_{(0)}=\mathcal{C}$ where $\mathcal{C}$ is the extremal binary self dual code of length 68 with parameters $\beta=103$ and $\gamma=4.$ Applying the $k^{th}$ range formula, we obtain:

\begin{table}[h!]
\caption{$i^{th}$ neighbour of  $\mathcal{N}_{(0)}$}\label{neighbors1}
\begin{center}\scalebox{0.9}{
\begin{tabular}{|c|c|ccc|}
\hline
$i$ & $\mathcal{N}_{(i+1)}$   & $x_i$  & $\gamma$ & $\beta$  \\ \hline \hline
$0$ & $\mathcal{N}_{(1)}$ &  $(1111011010011101111111100100111110)$  & $4$ & $101$   \\ \hline
$1$ & $\mathcal{N}_{(2)}$ &  $(0110100100111101111011111110111011)$  & $6$ & $145$   \\ \hline
$2$ & $\mathcal{N}_{(3)}$ &  $(0000100000010000011101110110000101)$  & $7$ & $152$   \\ \hline
$3$ & $\mathcal{N}_{(4)}$ &  $(1111111100000010000111001100101011)$  & $\textbf{7}$ & $\textbf{143}$   \\ \hline
$4$ & $\mathcal{N}_{(5)}$ &  $(0110010010100110110111101011111111)$  & $8$ & $162$   \\ \hline
$5$ & $\mathcal{N}_{(6)}$ &  $(1100001011011111001111110010001011)$  & $9$ & $174$   \\ \hline
$6$ & $\mathcal{N}_{(7)}$ &  $(1110010010100011111100101110001100)$  & $\textbf{9}$ & $\textbf{167}$   \\ \hline
$7$ & $\mathcal{N}_{(8)}$ &  $(0011000000000110110101001101100000)$  & $\textbf{9}$ & $\textbf{159}$   \\ \hline
$8$ & $\mathcal{N}_{(9)}$ &  $(1001101110001110110000111101000011)$  & $\textbf{9}$ & $\textbf{158}$   \\ \hline
$9$ & $\mathcal{N}_{(10)}$ &  $(1001011111100101110001001011110110)$  & $\textbf{9}$ & $\textbf{157}$   \\ \hline
$10$ & $\mathcal{N}_{(11)}$ &  $(1010101101101101110111011111111010)$  & $\textbf{9}$ & $\textbf{152}$   \\ \hline
$11$ & $\mathcal{N}_{(12)}$ &  $(1111010110110000110111011010101010)$  & $\textbf{7}$ & $\textbf{131}$   \\ \hline
$12$ & $\mathcal{N}_{(13)}$ &  $(1000011111111011110110001010110010)$  & $\textbf{6}$ & $\textbf{117}$   \\ \hline
\end{tabular}}
\end{center}
\end{table}

\newpage

We shall now separately consider the neighbours of $\mathcal{N}_{(7)}, \mathcal{N}_{(8)}, \mathcal{N}_{(10)}, \mathcal{N}_{(11)}, \mathcal{N}_{(12)}$ and $\mathcal{N}_{(13)}.$ We tabulate the results below. All the codes in Table~4 have an automorphism group of order 1.

\begin{table}[h!]\caption{New codes of length 68 as neighbors}\label{neighbors}
\begin{center}\scalebox{0.7}{
\begin{tabular}{|c|c|ccc|c|c|ccc|}
\hline
$\mathcal{N}_{(i)}$ & $\mathcal{M}_{i}$ & $(x_{35},x_{36},...,x_{68})$  & $\gamma$ & $\beta$  &
$\mathcal{N}_{(i)}$ & $\mathcal{M}_{i}$ & $(x_{35},x_{36},...,x_{68})$  & $\gamma$ & $\beta$  \\ \hline
$7$ & $$ & $(0001000010111101010000011101000110)$  & $\textbf{7}$ & $\textbf{141}$   &
$7$ & $$ & $(0100101001111001101010101010101110)$  & $\textbf{8}$ & $\textbf{150}$   \\ \hline
$7$ & $$ & $(0111001010000000100011000001011100)$  & $\textbf{8}$ & $\textbf{151}$   &
$7$ & $$ & $(1001100101100110101111100011101101)$  & $\textbf{8}$ & $\textbf{152}$   \\ \hline
$7$ & $$ & $(0011100111101011010101111011100100)$  & $\textbf{9}$ & $\textbf{164}$   &
$7$ & $$ & $(1000000010011000001010001011010011)$  & $\textbf{9}$ & $\textbf{165}$   \\ \hline
$7$ & $$ & $(0010010111100000100111110000000000)$  & $\textbf{9}$ & $\textbf{166}$   &
$7$ & $$ & $(0010101001010010101010100000000011)$  & $\textbf{9}$ & $\textbf{168}$   \\ \hline
$7$ & $$ & $(1000101001011010000100100100010010)$  & $\textbf{9}$ & $\textbf{170}$   &
$7$ & $$ & $(0110110001000000000110000010011110)$  & $\textbf{9}$ & $\textbf{172}$   \\ \hline
\end{tabular}}
\end{center}

\begin{center}\scalebox{0.7}{
\begin{tabular}{|c|c|ccc|c|c|ccc|}
\hline
$\mathcal{N}_{(i)}$ & $\mathcal{M}_{i}$ & $(x_{35},x_{36},...,x_{68})$  & $\gamma$ & $\beta$ &
$\mathcal{N}_{(i)}$ & $\mathcal{M}_{i}$ & $(x_{35},x_{36},...,x_{68})$  & $\gamma$ & $\beta$  \\ \hline
$8$ & $$ & $(0111100101101011111001111110111101)$  & $\textbf{7}$ & $\textbf{134}$   &
$8$ & $$ & $(1000001110101000000101110110100010)$  & $\textbf{7}$ & $\textbf{135}$   \\ \hline
$8$ & $$ & $(1111010110000000111001101001000000)$  & $\textbf{7}$ & $\textbf{136}$   &
$8$ & $$ & $(1111000111011000110111001101111110)$  & $\textbf{7}$ & $\textbf{137}$   \\ \hline
$8$ & $$ & $(0010000011001100110010010001100001)$  & $\textbf{7}$ & $\textbf{138}$   &
$8$ & $$ & $(1111001001110111001001100101001100)$  & $\textbf{7}$ & $\textbf{139}$   \\ \hline
$8$ & $$ & $(1011011001100110111011100100011000)$  & $\textbf{8}$ & $\textbf{144}$   &
$8$ & $$ & $(1111101110010110001101111111010010)$  & $\textbf{8}$ & $\textbf{147}$   \\ \hline
$8$ & $$ & $(0011000110101010001011010101100101)$  & $\textbf{8}$ & $\textbf{148}$   &
$8$ & $$ & $(0110110000000110010110011110100110)$  & $\textbf{8}$ & $\textbf{149}$   \\ \hline
$8$ & $$ & $(0001101100111000101110011001001001)$  & $\textbf{9}$ & $\textbf{160}$   &
$8$ & $$ & $(1000001000111101010110000101010001)$  & $\textbf{9}$ & $\textbf{161}$   \\ \hline
$8$ & $$ & $(1110100010110010110000010010000101)$  & $\textbf{9}$ & $\textbf{162}$   &
$8$ & $$ & $(0100011010001111001111101001011111)$  & $\textbf{9}$ & $\textbf{163}$   \\ \hline

\end{tabular}}
\end{center}

\begin{center}\scalebox{0.7}{
\begin{tabular}{|c|c|ccc|c|c|ccc|}
\hline
$\mathcal{N}_{(i)}$ & $\mathcal{M}_{i}$ & $(x_{35},x_{36},...,x_{68})$  & $\gamma$ & $\beta$ &
$\mathcal{N}_{(i)}$ & $\mathcal{M}_{i}$ & $(x_{35},x_{36},...,x_{68})$  & $\gamma$ & $\beta$ \\ \hline
$10$ & $$ & $(1101100111000110001101001101111000)$  & $\textbf{7}$ & $\textbf{132}$   &
$10$ & $$ & $(1111000110101101101011011011000011)$  & $\textbf{8}$ & $\textbf{143}$   \\ \hline
$10$ & $$ & $(1110110011011110001010110001101011)$  & $\textbf{8}$ & $\textbf{145}$   &
$10$ & $$ & $(0010111101110011010001011100111110)$  & $\textbf{8}$ & $\textbf{146}$   \\ \hline
$10$ & $$ & $(1011010011010100010100010010111010)$  & $\textbf{9}$ & $\textbf{156}$   &
$$ & $$ & $$  & $\textbf{}$ & $\textbf{}$   \\ \hline
\end{tabular}}
\end{center}

\begin{center}\scalebox{0.7}{
\begin{tabular}{|c|c|ccc|c|c|ccc|}
\hline
$\mathcal{N}_{(i)}$ & $\mathcal{M}_{i}$ & $(x_{35},x_{36},...,x_{68})$  & $\gamma$ & $\beta$ &
$\mathcal{N}_{(i)}$ & $\mathcal{M}_{i}$ & $(x_{35},x_{36},...,x_{68})$  & $\gamma$ & $\beta$\\ \hline
$11$ & $$ & $(0101000101100110011001011000111100)$  & $\textbf{8}$ & $\textbf{139}$   &
$11$ & $$ & $(1011111100100001110111000101111100)$  & $\textbf{8}$ & $\textbf{140}$   \\ \hline
$11$ & $$ & $(1100011001000111000000110111010110)$  & $\textbf{8}$ & $\textbf{141}$   &
$11$ & $$ & $(1101111110110100001101111111011101)$  & $\textbf{9}$ & $\textbf{151}$   \\ \hline
$11$ & $$ & $(0001100111110011010110111001111010)$  & $\textbf{9}$ & $\textbf{154}$   &
$11$ & $$ & $(0100100111101001001010101111000001)$  & $\textbf{9}$ & $\textbf{155}$   \\ \hline
\end{tabular}}
\end{center}

\begin{center}\scalebox{0.7}{
\begin{tabular}{|c|c|ccc|c|c|ccc|}
\hline
$\mathcal{N}_{(i)}$ & $\mathcal{M}_{i}$ & $(x_{35},x_{36},...,x_{68})$  & $\gamma$ & $\beta$ &
$\mathcal{N}_{(i)}$ & $\mathcal{M}_{i}$ & $(x_{35},x_{36},...,x_{68})$  & $\gamma$ & $\beta$ \\ \hline
$12$ &  & $(1100011100101100111101111001101100)$  & $\textbf{6}$ & $\textbf{121}$   &
$12$ & $$ & $(1111100111100011111001011110101111)$  & $\textbf{6}$ & $\textbf{123}$   \\ \hline
$12$ & $$ & $(0001101000001011101010000001100001)$  & $\textbf{6}$ & $\textbf{124}$   &
$12$ & $$ & $$  & $\textbf{}$ & $\textbf{}$   \\ \hline
\end{tabular}}
\end{center}

\begin{center}\scalebox{0.7}{
\begin{tabular}{|c|c|ccc|c|c|ccc|}
\hline
$\mathcal{N}_{(i)}$ & $\mathcal{M}_{i}$ & $(x_{35},x_{36},...,x_{68})$  & $\gamma$ & $\beta$ &
$\mathcal{N}_{(i)}$ & $\mathcal{M}_{i}$ & $(x_{35},x_{36},...,x_{68})$  & $\gamma$ & $\beta$\\ \hline
$13$ & $$ & $(0101110011001101000001001000001000)$  & $\textbf{5}$ & $\textbf{110}$   &
$13$ & $$ & $(1011011100110111011001011010101001)$  & $\textbf{6}$ & $\textbf{120}$   \\ \hline
$13$ &  & $(1000111010010011011000110000101011)$  & $\textbf{6}$ & $\textbf{122}$   &
$13$ & $$ & $$  & $\textbf{}$ & $\textbf{}$   \\ \hline
\end{tabular}}
\end{center}
\end{table}

\newpage

As we can see, we were able to construct many extremal binary self-dual codes of length 68 with new weight enumerators for the rare parameters $\gamma=7,8$ and $9.$
 
\newpage

\section{Conclusion}

In this paper, we have extended the idea of $G$-codes to composite $G$-codes. We have shown that similarly as the $G$-codes, the composite $G$-codes are also ideals in the group ring $RG.$ We have shown that the dual of a composite $G$-code is also a $G$-code. We have studied self-orthogonal and self-dual composite $G$-codes over rings. Moreover, we have extended the results on quasi-$G$-codes to quasi-composite $G$-codes. We have also generalized results on self-dual codes obtained from generator matrices of the form $[I \ | \ \Omega(v)],$ where $\Omega(v)$ is the composite matrix. Additionally in this work, we were able to construct the following extremal binary self-dual codes with new weight enumerators in $W_{68,2}$:

\begin{equation*}
\begin{split}
(\gamma =5,& \quad \beta =\{110\}). \\
(\gamma =6,& \quad \beta =\{117,120,121,122,123,124\}). \\
(\gamma =7,& \quad \beta =\{131,132,134,135,136,137,138,139,141,143\}). \\
(\gamma =8,& \quad \beta =\{139,140,141,143,144,145,146,147,148,149,150,151,152\}).\\
(\gamma =9,& \quad \beta =\{151,152,154,155,156,157,158,159,160,161,162,163,164,165,\\ & \qquad \quad 166,167,168,170,172\}). \\
\end{split}%
\end{equation*}

A suggestion for future work would be to consider composite matrices of greater lengths to search for extremal binary self dual codes over different rings. Another direction is to determine which codes are composite $G$-codes for a finite group $G.$

\end{document}